\documentclass[letterpaper, 10 pt, conference]{ieeeconf}  

\IEEEoverridecommandlockouts                              

\usepackage{graphicx} 
\usepackage{amsmath} 
\usepackage{amssymb}  
\usepackage{mathrsfs}
\usepackage{xcolor}
\usepackage{physics}
\newtheorem{proposition}{Proposition}
\newtheorem{corollary}[proposition]{Corollary}

\newtheorem{assumption}{Assumption}
\newtheorem{definition}{Definition}

\title{\LARGE \bf Neural Robust Control on Lie Groups Using Contraction Methods (Extended Version)}

\author{Yi Lok Lo$^{1}$, Longhao Qian$^{1}$ and Hugh H.T. Liu$^{1}$
\thanks{$^{1}$The authors are with the Institute for Aerospace Studies, University of Toronto, 4925 Dufferin St, North York, ON M3H 5T5, Canada 
{\tt\small enoch.lo@mail.utoronto.ca, longhao.qian@mail.utoronto.ca, hugh.liu@utoronto.ca}}%
\thanks{This work has been submitted to the IEEE for possible publication. 
Copyright may be transferred without notice, after which this version may no longer be accessible.}
}

\begin{document}

\maketitle
\thispagestyle{empty}
\pagestyle{empty}

\begin{abstract}

In this paper, we propose a learning framework for synthesizing a robust controller for dynamical systems evolving on a Lie group. A robust control contraction metric (RCCM) and a neural feedback controller are jointly trained to enforce contraction conditions on the Lie group manifold. Sufficient conditions are derived for the existence of such an RCCM and neural controller, ensuring that the geometric constraints imposed by the manifold structure are respected while establishing a disturbance-dependent tube that bounds the output trajectories. As a case study, a feedback controller for a quadrotor is designed using the proposed framework. Its performance is evaluated using numerical simulations and compared with a geometric controller.

\end{abstract}



\section{INTRODUCTION}

Contraction analysis provides a powerful framework for studying the incremental stability of nonlinear dynamical systems, establishing global stability properties through differential dynamics \cite{Lohmiller1998}. Building on this foundation, control contraction metrics (CCMs) provide constructive conditions for synthesizing stabilizing feedback controllers for nonlinear systems \cite{Manchester2017}. To address robustness, robust control contraction metrics (RCCMs) extend the CCM framework to systems subject to disturbances. Prior works have established bounds on $\mathcal{L}_\infty$ and $\mathcal{L}_2$ gains between disturbances and trajectory deviations \cite{Zhao2022,Manchester2018}. Furthermore, disturbance estimation techniques have been integrated with contraction-based methods to learn and compensate for model uncertainties \cite{Zhao2024}.

In parallel, rapid advances in machine learning have enabled learning-based approaches for synthesizing controllers for complex nonlinear systems using reinforcement learning \cite{Brunke2022} and deep neural networks \cite{Dawson2023}. In particular, several learning-based CCM frameworks have been developed, offering interpretable and certifiable guarantees \cite{Tsukamoto2021, Sun2021, Tsukamoto2020}.

Existing CCM and RCCM formulations are primarily developed for systems in Euclidean spaces, yet many robotic applications contain states evolving on manifolds, such as attitudes of aerial vehicles or joint configurations of manipulators. A common workaround uses local coordinates, such as Euler angles \cite{Sun2021, Singh2023}, but introduces singularities and distortions that can degrade control performance. In contrast, Lie group representations provide globally valid models. Recent works have extended contraction control to submanifolds embedded in Euclidean spaces \cite{Wu2024, Vang2020}. In this paper, we present an approach towards certified robust trajectory tracking for systems evolving on Lie groups using a learning framework. The contributions of this paper are summarized as follows:

\begin{enumerate}
    \item We formulate conditions for RCCMs on Lie groups such that the closed-loop system exhibits input-output stability with a bounded $\mathcal{L}$-infinity gain. Conditions for the existence of RCCM controllers are provided. 
    \item We develop a learning-based framework that jointly synthesizes a neural feedback controller and an RCCM on Lie groups. A quadrotor is used as a case study, and numerical simulations demonstrate robust trajectory tracking capabilities of the learned controller.
\end{enumerate}

\textit{Notations:} Given a manifold $\mathcal X$, $T\mathcal X$ and $T_{\boldsymbol{x}}\mathcal X$ represent the tangent bundle and tangent space at $\boldsymbol{x}\in\mathcal{X}$ respectively. $\boldsymbol{I}_n$ and $\boldsymbol{0}$ denote an $n\times n$ identity matrix and a zero matrix of compatible dimensions respectively. For symmetric matrices $\boldsymbol{P}$ and $\boldsymbol{Q}$, $\boldsymbol{P}\succ \boldsymbol{Q}$ ($\boldsymbol{P}\succeq \boldsymbol{Q}$) means $\boldsymbol{P} - \boldsymbol{Q}$ is positive definite (semi-definite). $\langle \boldsymbol{A}\rangle$ denotes $\boldsymbol{A}+\boldsymbol{A}^\top$ for a square matrix $\boldsymbol{A}$, and $\operatorname{vec}(
\boldsymbol{B})$ denotes the row-major vectorization of a matrix $\boldsymbol{B}$. $\operatorname{diag}(\cdot)$ denotes the (block-)diagonal matrix operator. Let $\boldsymbol{\phi}^{\wedge}$ be the skew-symmetric operator for $\boldsymbol{\phi} \in \mathbb{R}^{3}$ such that ${\boldsymbol{\phi}^{\wedge}}^\top=-\boldsymbol{\phi}^{\wedge}$. Let $\boldsymbol{S}^{\vee}$ be the vee-map for a skew-symmetric matrix $\boldsymbol{S}^\top = -\boldsymbol{S} \in \mathbb{R}^{3 \times 3}$. For a smooth function $\boldsymbol{Q}(\boldsymbol{x})$, $\partial_{\boldsymbol{f}}\boldsymbol{Q}(\boldsymbol{x})=\sum_j\frac{\partial \boldsymbol{Q}}{\partial x_j}f_j$ denotes its directional derivative along vector field $\boldsymbol{f}$. The Euclidean 2-norm is denoted as $\|\cdot\|$, and the $\mathcal{L}_\infty$ norm of a function $\boldsymbol{x}(t)$ is defined as $\|\boldsymbol{x}(t)\|_{\mathcal{L}_{
\infty}} = \sup_{t\geq0} \| \boldsymbol{x}(t)\|$. The truncated $\mathcal{L}_\infty$ norm is defined as $\|\boldsymbol{x}(t)\|_{\mathcal{L}_{
\infty}^{[0, T]}} = \sup_{0\leq t\leq T} \| \boldsymbol{x}(t)\|$. For a positive integer $m$, denote the set $\mathbb Z_m := \{1, ..., m\}$. The vectors $\boldsymbol{e}_i$, for $i\in \mathbb Z_3$, represent the Euclidean basis in $\mathbb R^3$. 

\section{Problem statement and Preliminaries}

Consider the following nonlinear control-affine system
\begin{equation}\label{eqn: sys_ctrl_affine_dist}
\begin{aligned}
    &\dot{\boldsymbol{x}} = \boldsymbol{f}(\boldsymbol{x}) + \boldsymbol{B}(\boldsymbol{x})\boldsymbol{u} + \boldsymbol{B}_w(\boldsymbol{x})\boldsymbol{w}\\
    &\boldsymbol{z} = \boldsymbol{g}(\boldsymbol{x}, \boldsymbol{u})
\end{aligned}
\end{equation}
with the state $\boldsymbol{x}\in\mathcal X\subseteq \mathbb R^n$, control $\boldsymbol{u}\in\mathbb R^m$, disturbance $\boldsymbol{w}\in\mathbb R^p$ and output $\boldsymbol{z}\in \mathbb R^l$. $\boldsymbol{f}$, $\boldsymbol{B}$ and $\boldsymbol{B}_w$
are continuously differentiable vector/matrix functions with compatible dimensions, and matrices $\boldsymbol{B}(\boldsymbol{x})$ and $\boldsymbol{B}_w(\boldsymbol{x})$ have full column rank. We use $\boldsymbol{b}_i$ and $\boldsymbol{b}_{w,j}$ to denote the $i^{th}$ column of $\boldsymbol{B}(\boldsymbol{x})$ and $j^{th}$ column of $\boldsymbol{B}_w(\boldsymbol{x})$ respectively. The state space $\mathcal{X}$ is assumed to be a Lie group, which can be viewed as a submanifold embedded in a higher-dimensional Euclidean space (e.g. $SO(n) = \{\boldsymbol{R}\in \mathbb R^{n\times n} \mid \boldsymbol{R}^\top\boldsymbol{R}=\boldsymbol{I}_n \,, \det\boldsymbol{R} = 1\}$). The state space of the system can be represented as
\begin{equation}
    \mathcal{X} = \{\boldsymbol{x}\in \mathbb R^{n} \mid \boldsymbol{h}(\boldsymbol{x})=\boldsymbol{0} \},
\end{equation}
where $\boldsymbol{h}:\mathbb R^n \to \mathbb R^{n-q}$ is a smooth constraint function. Hence, $\mathcal X$ can be regarded as a $q$-dimensional manifold embedded in $\mathbb R^n$. To ensure trajectories remain on $\mathcal{X}$, it is necessary and sufficient that the vector fields $\boldsymbol{f}$, $\boldsymbol{b}_i\,(i\in \mathbb{Z}_m)$, $\boldsymbol{b}_{w,j}\,(j\in \mathbb{Z}_p)$ be tangent to the manifold $\mathcal{X}$. Equivalently, they must satisfy the transversality conditions. For all $\boldsymbol{x}\in \mathcal{X}$,
\begin{equation}\label{eqn: cond_transversality}
    \partial_{\boldsymbol{f}}\boldsymbol{h}(\boldsymbol{x}) = \boldsymbol{0}, \;\partial_{\boldsymbol{b}_i}\boldsymbol{h}(\boldsymbol{x}) = \boldsymbol{0}, \;\partial_{\boldsymbol{b}_{w,j}}\boldsymbol{h}(\boldsymbol{x}) = \boldsymbol{0}.
\end{equation}

It is well-known that any Lie group is parallelizable, admitting a set of global left-invariant vector fields that trivializes its tangent bundle \cite{Lee2003}. Therefore, there exist global vector fields $\{\boldsymbol{s}_1, ..., \boldsymbol{s}_q\}$ that form a basis of $T_{\boldsymbol{x}}\mathcal{X}$ at each $\boldsymbol{x}\in\mathcal{X}$. For Lie groups that can be characterized in the general form $G_c\times G_m\times \mathbb R^{a}$, where $G_c$ is a compact Lie group and $G_m$ is a matrix Lie group, a smooth function $\boldsymbol{S}(\boldsymbol{x}):=\begin{bmatrix}
    \boldsymbol{s}_1, ..., \boldsymbol{s}_q
\end{bmatrix} \in \mathbb R^{n\times q}$ can be defined such that 
\begin{equation}\label{eqn: extrinsic_to_intrinsic}
    \forall\boldsymbol{x}\in\mathcal X,\,\,T_{\boldsymbol{x}}\mathcal X=\{ \boldsymbol{S}(\boldsymbol{x})\boldsymbol{v}\mid\boldsymbol{v}\in \mathbb{R}^q\} 
\end{equation}
That is, any tangent vector at $\boldsymbol{x}$ with extrinsic representation (written in coordinates of ambient Euclidean space) can be obtained from linear combinations of tangent basis vector $\boldsymbol{S}(\boldsymbol{x})\boldsymbol{v}$, where $\boldsymbol{v}$ is the intrinsic representation of said tangent vector.

\begin{assumption}\label{assum: sys_lie}
    The system defined in \eqref{eqn: sys_ctrl_affine_dist} evolves on the Lie group $\mathcal X$ in the general form $G_c\times G_m\times \mathbb R^{a}$ such that the transversality conditions in \eqref{eqn: cond_transversality} are satisfied at all times. Moreover, the tangent bundle $T\mathcal X$ is spanned by smooth vector fields $\boldsymbol{S}(\boldsymbol{x})$,
    and $\boldsymbol{S}^\top\boldsymbol{S}$ is uniformly bounded.
\end{assumption}

For systems that satisfy Assumption \ref{assum: sys_lie}, given a nominal solution $(\boldsymbol{x}^*, \boldsymbol{u}^*, \boldsymbol{w}^*, \boldsymbol{z}^*)$ that satisfies \eqref{eqn: sys_ctrl_affine_dist}, this paper focuses on designing a feedback controller in the form 
\begin{equation}\label{eqn: fb_ctrl_format}
    \boldsymbol{u} = \boldsymbol{k}(\boldsymbol{x}, \boldsymbol{x}^*) + \boldsymbol{u}^*, 
\end{equation}
where $\boldsymbol{k}:\mathcal{X}\times\mathcal{X}\to\mathbb{R}^m$ and $\boldsymbol{k}(\boldsymbol{x}, \boldsymbol{x})  =\boldsymbol{0}$. The goal is to minimize the universal $\mathcal{L}_\infty$ gain from disturbance deviation $\boldsymbol{w}-\boldsymbol{w}^*$ to output deviation $\boldsymbol{z}-\boldsymbol{z}^*$. Note that $\boldsymbol{w}^*$ is the nominal disturbance vector including the special case $\boldsymbol{w}^*(t)=\boldsymbol{0}$. In practice, $\boldsymbol{w}^*(t)$ can be disturbance estimates from some disturbance observer. We define the universal $\mathcal{L}_\infty$ gain formally in Definition \ref{defn: universal_L_inf_gain}. 

\begin{definition}\label{defn: universal_L_inf_gain}
A control system \eqref{eqn: sys_ctrl_affine_dist} with feedback controller \eqref{eqn: fb_ctrl_format} achieves a $\mathcal{L}_\infty$ gain bound of $\alpha$, if for any nominal trajectory $\boldsymbol{x}^*,\boldsymbol{w}^*,\boldsymbol{z}^*$ satisfying \eqref{eqn: sys_ctrl_affine_dist}, any initial condition $\boldsymbol{x}(0)$, any $\boldsymbol{w}$ such that $\|\boldsymbol{w} \|_{\mathcal{L}_{\infty}}<\infty$ and for any $T>0$, the condition
\begin{equation}
    \|\boldsymbol{z}-\boldsymbol{z}^*\|^2_{\mathcal{L}_\infty ^{[0,T]}} \leq \alpha^2\|\boldsymbol{w}-\boldsymbol{w}^*\|^2_{\mathcal{L}_\infty ^{[0,T]}} + \beta(\boldsymbol{x}(0), \boldsymbol{x}^*(0))
\end{equation}
holds for a function $\beta(\boldsymbol{x}_1, \boldsymbol{x}_2) \geq 0$ and $\beta(\boldsymbol{x}, \boldsymbol{x}) = 0$.
\end{definition}

In the case where $\boldsymbol{z}$ lives in the Lie group $\mathcal{Z}\subseteq\mathcal{X}\times\mathbb R^m$, output deviations can be measured using Riemannian distances under some Riemannian metrics on $\mathcal{Z}$.

\subsection{Preliminaries}
The CCM approach is a framework for designing feedback controllers for control-affine systems that provide incremental stability guarantees through analysis of the associated differential dynamics \cite{Manchester2017}. When disturbances are considered, the RCCM framework proposed in \cite{Manchester2018} further analyzes the differential dynamics associated with \eqref{eqn: sys_ctrl_affine_dist}, which is given by
\begin{equation}\label{eqn: sys_differential_dist}
\begin{aligned}
    &\delta\dot{\boldsymbol{x}} = \boldsymbol{A}(\boldsymbol{x}, \boldsymbol{u}, \boldsymbol{w})\delta\boldsymbol{x} + \boldsymbol{B}(\boldsymbol{x})\delta\boldsymbol{u} + \boldsymbol{B}_w(\boldsymbol{x})\delta\boldsymbol{w} \\
    &\delta\boldsymbol{z} = \boldsymbol{C}(\boldsymbol{x},\boldsymbol{u})\delta\boldsymbol{x} + \boldsymbol{D}(\boldsymbol{x},\boldsymbol{u})\delta\boldsymbol{u}
\end{aligned}
\end{equation}
where $\boldsymbol{A}(\boldsymbol{x}, \boldsymbol{u}, \boldsymbol{w}) = \frac{\partial \boldsymbol{f}}{\partial \boldsymbol{x}} + \sum_{i=1}^m \frac{\partial \boldsymbol{b}_i}{\partial \boldsymbol{x}} u_{i} + \sum_{j=1}^p \frac{\partial \boldsymbol{b}_{w,j}}{\partial \boldsymbol{x}} w_{j}$, $\boldsymbol{C}(\boldsymbol{x},\boldsymbol{u}) = \frac{\partial \boldsymbol{g}}{\partial \boldsymbol{x}}$ and $\boldsymbol{D}(\boldsymbol{x},\boldsymbol{u}) = \frac{\partial \boldsymbol{g}}{\partial \boldsymbol{u}}$, with $u_i$ and $w_j$ denoting the $i^{th}$ element of $\boldsymbol{u}$ and $j^{th}$ element of $\boldsymbol{w}$ respectively. 

The RCCM approach for a system with state vectors that live in Euclidean space ($\mathcal{X}=\mathbb R^n$) involves finding a Riemannian metric $\boldsymbol{M}(\boldsymbol{x}):\mathbb{R}^n\to\mathbb{R}^{n\times n}$ and feedback controller in the form of \eqref{eqn: fb_ctrl_format} such that $\boldsymbol{M}(\boldsymbol{x})$ is uniformly bounded (i.e., $\underline{m}\boldsymbol{I}_n\preceq\boldsymbol{M}(\boldsymbol{x})\preceq\overline{m}\boldsymbol{I}_n$) and the inner product of tangent vectors with respect to the metric $V(\boldsymbol{x}, \delta\boldsymbol{x})=\delta\boldsymbol{x}^{\top}\boldsymbol{M}(\boldsymbol{x})\delta\boldsymbol{x}$ decays to some disturbance-dependant bound. Under such conditions, the metric $\boldsymbol{M}$ is termed an RCCM, and its associated feedback controller will then guarantee $H_\infty$-like input-output stability for the closed-loop system \cite{Manchester2018}. 

However, for states that live in a general Lie group, constraints are posed on both the states and differential states, such that $(\boldsymbol{x},\delta\boldsymbol{x})\in T\mathcal{X}$. \cite{Wu2024} laid the groundwork to search for a CCM that respects the inherent geometry of the underlying submanifold $\mathcal X$. First, tangent vectors of the system can be represented with intrinsic coordinates following \eqref{eqn: extrinsic_to_intrinsic}. Therefore, for systems that satisfy Assumption \ref{assum: sys_lie}, we are able to find smooth functions $\boldsymbol{v}\in \mathbb R^q$, $\boldsymbol{E}(\boldsymbol{x})\in \mathbb R^{q\times m}$ and $\boldsymbol{E}_w(\boldsymbol{x})\in \mathbb R^{q\times p}$ such that 
\begin{equation}\label{eqn: E_matrix_transformation}
    \delta\boldsymbol{x} = \boldsymbol{S}\boldsymbol{v},\,\boldsymbol{B}=\boldsymbol{S}\boldsymbol{E},\,\boldsymbol{B}_w=\boldsymbol{S}\boldsymbol{E}_w.
\end{equation}
Moreover, instead of finding a CCM $\boldsymbol{M}(\boldsymbol{x})\in\mathbb R^{n\times n}$ that respects the geometry of $\mathcal{X}$, it is proposed to search for a lower-dimensional metric $\boldsymbol{\mathcal M}(\boldsymbol{x}): \mathcal{X} \to \mathbb R ^{q\times q}$ and project it to its full dimensions using $\boldsymbol{M}(\boldsymbol{x}) = \boldsymbol{P}_S(\boldsymbol{x}) \boldsymbol{\mathcal M}(\boldsymbol{x}) \boldsymbol{P}_S^\top(\boldsymbol{x})$,
where $\boldsymbol{P}_S=\boldsymbol{S}(\boldsymbol{S}^\top\boldsymbol{S})^{-1}$ is the projection operator. The existence and smoothness of $\boldsymbol{P}_S(\boldsymbol{x})$ are guaranteed by Assumption \ref{assum: sys_lie}. $\boldsymbol{\mathcal M}(\boldsymbol{x})$ should be uniformly bounded such that there exist two positive constants $\underline{\mathfrak m}, \overline{\mathfrak m}$ satisfying
\begin{equation}\label{eqn: cond_metric_uni_bound}
    \underline{\mathfrak m}\boldsymbol{I}_q \preceq \boldsymbol{\mathcal M}(\boldsymbol{x}) \preceq \overline{\mathfrak m}\boldsymbol{I}_q.
\end{equation}

With this method, the tangent vectors respect the constraints from the tangent bundle and the Riemannian metric remains a smooth varying inner product on the tangent space of the state manifold $\mathcal X$, such that $\forall (\boldsymbol{x},\,\boldsymbol{S}(\boldsymbol{x})\boldsymbol{v})\in T\mathcal{X}$
\begin{equation}\label{eqn: intrinsic_diff_CLF}
    V = \boldsymbol{v}^\top\boldsymbol{S}^\top\boldsymbol{P}_S\boldsymbol{\mathcal M}\boldsymbol{P}_S^\top\boldsymbol{S}\boldsymbol{v} = \boldsymbol{v}^\top\boldsymbol{\mathcal M}\boldsymbol{v}.
\end{equation}

Some notations on Riemannian geometry are introduced before moving on to the main results. Given a Lie group state space manifold $\mathcal X$, the set of all possible paths between $a$ and $b$ along $\mathcal X$ is denoted as $\Gamma_{\mathcal X}(a, b)$, where each $\boldsymbol{c}\in \Gamma_{\mathcal X}(a, b)$ represents a smooth mapping $\boldsymbol{c}(s):[0,1]\to\mathcal{X}$ satisfying $\boldsymbol{c}(0)=a$, $\boldsymbol{c}(1)=b$. Hence, for all $s\in [0,1]$, $\boldsymbol{c}(s)\in\mathcal{X}$ and $\frac{\partial \boldsymbol{c}}{\partial s}(s)=\boldsymbol{S}(\boldsymbol{c}(s))\boldsymbol{v}\in T_{\boldsymbol{c}(s)}\mathcal{X}$. Given the state manifold $\mathcal X$ equipped with a Riemannian metric $\boldsymbol{\mathcal{M}}$, the Riemannian energy of a path $\boldsymbol{c}$ is defined as $\mathcal{E}(\boldsymbol{c})=\int^1_0\boldsymbol{v}^\top\boldsymbol{\mathcal{M}}(\boldsymbol{c}(s))\boldsymbol{v}\,ds$. The notation $\mathcal{E}(a,b)$ also denotes the minimum energy of a path joining $a,b$ such that $\mathcal{E}(a,b)=\inf_{\boldsymbol{c}\in\Gamma_{\mathcal{X}}(a,b)}\mathcal{E}(\boldsymbol{c})$.

Some preliminaries on Matrix Lie groups are also provided. We define $\operatorname{Tr}(\cdot)$ to be the matrix trace and $\otimes$ as the Kronecker product. $\cdot^\wedge$ and $\cdot^\vee$ denote hat-map and vee-map respectively, which maps a Euclidean vector to its Lie algebra linearly, and vice versa. We consider Lie groups in the form $\mathbb R^{a}\times G_c\times G_m$, where $G_c$ is a compact Lie group and $G_m$ is a matrix Lie group. As compact Lie groups are isomorphic to matrix Lie groups, given a matrix Lie group $\mathcal G$ with matrix elements $\boldsymbol{G}\in\mathcal G$ embedded in $\mathbb R^{d\times d}$, representing a dimension $q_G$ manifold, there exists an associated Lie algebra $\mathfrak g$ diffeomorphic to $\mathbb R^{q_G}$. We then define a matrix $\boldsymbol{E}_{vec} \in \mathbb R^{d^2\times q_G}$ such that each column $\boldsymbol{E}_{vec,i}$ represents the vectorized basis vector 
\begin{equation}
    \boldsymbol{E}_{vec,i}=\operatorname{vec}(\boldsymbol{e}_i^\wedge)
\end{equation}
for all $i\in\mathbb Z_{q_G}$. Therefore, for each $\boldsymbol{\Xi}\in\mathfrak g$, we can find $\boldsymbol{v}\in \mathbb R^{q_G}$ such that $\operatorname{vec}(\boldsymbol{\Xi})=\boldsymbol{E}_{vec}\boldsymbol{v}$, where $\boldsymbol{v}$ represents the left-trvialized intrinsic coordinates of tangent vector as in ([Sv]). Therefore, $\boldsymbol{S}$ matrix for matrix Lie groups can be found as 
\begin{equation}
\begin{aligned}
    \operatorname{vec}(\dot{\boldsymbol{G}}) &= \operatorname{vec}(\boldsymbol{G}\boldsymbol{\Xi})=(\boldsymbol{I}_{d}\otimes\boldsymbol{G})\operatorname{vec}(\boldsymbol{\Xi})\\
    &=(\boldsymbol{I}_{d}\otimes\boldsymbol{G})\boldsymbol{E}_{vec}\boldsymbol{v} = \boldsymbol{S}\boldsymbol{v},
\end{aligned}
\end{equation}
using the identity vectorization identities. Therefore, each column of $\boldsymbol{S}$ matrix $\boldsymbol{s}_i$ can be found as $\boldsymbol{s}_i = \operatorname{vec}(\boldsymbol{G}\boldsymbol{e}_i^\wedge)$ for all $i\in \mathbb Z_{q_G}$. Moreover, 

For the Euclidean space $\mathbb R^{a}$, the $\boldsymbol{S}$ and $\boldsymbol{E}_{vec}$ matrix and is just the identity $\boldsymbol{I}_{a}$ as the hat map in Euclidean space is trivially $\boldsymbol{e}_i^\wedge = \boldsymbol{e}_i$. For a state space $\mathcal{X}=\mathcal{G}_1\times ... \times\mathcal{G}_k\times\mathbb R^{a}$ composed of multiple compact Lie groups and the Euclidean, the $\boldsymbol{S}$ and $\boldsymbol{E}_{vec}$ can be built by block diagonalization of each of their respective tangent vector span and basis vector span matrices. That is, $\boldsymbol{S}=\operatorname{diag}(\boldsymbol{S}_1, ..., \boldsymbol{S}_k,\boldsymbol{I}_a)$.

Moreover, a left-invariant metric tensor $g_{\mathcal{G}}(\cdot,\cdot):\mathfrak{g}\times\mathfrak{g}\to\mathbb R$ can be characterized on a matrix Lie group $\mathcal G$ using Lie algebra, where $g_{\mathcal{G}}(\boldsymbol{\Xi}_1, \boldsymbol{\Xi}_2)=\operatorname{Tr}(\boldsymbol{\Xi}_1^\top\boldsymbol{\Xi}_2)=\operatorname{vec}(\boldsymbol{\Xi}_1)^\top\operatorname{vec}(\boldsymbol{\Xi}_2)=\boldsymbol{v}_1^\top\boldsymbol{E}_{vec}^\top\boldsymbol{E}_{vec}\boldsymbol{v}_2$. Therefore, given some paths $\boldsymbol{c}(s)\in \Gamma_{\mathcal{G}}(a,b)$, the length can be defined as $L_{\mathcal{G}}(a,b)=\int^1_0\sqrt{\boldsymbol{v}(s)^\top\boldsymbol{E}_{vec}^\top\boldsymbol{E}_{vec}\boldsymbol{v}(s)} ds$ and the geodesic is defined as $\boldsymbol{\gamma}_{\mathcal{G}}(t)= \arg_{\boldsymbol{c}\in\Gamma_{\mathcal{G}}(a,b)}\min{L(\boldsymbol{c})}$. This works trivially for Euclidean space as well with $\boldsymbol{E}_{vec}=\boldsymbol{I}_{a}$.

\section{Neural Robust Control on Lie Groups Using Contraction Methods}
In this section, we present our main results. Conditions for the closed-loop system to achieve a given $\mathcal{L}_\infty$ gain bound are first proposed. Then, a learning-based framework is implemented to find an RCCM and a feedback controller that satisfy these conditions and minimize the gain bound. 

\subsection{RCCMs on Lie groups}

\cite{Zhao2022} established linear matrix inequalities (LMIs) guaranteeing a $\mathcal{L}_\infty$ gain bound for Euclidean systems. We extend these results to systems evolving on Lie groups. 

We consider the differential closed-loop system associated with \eqref{eqn: sys_ctrl_affine_dist} expressed in intrinsic coordinates $\boldsymbol{v}$. Given a feedback controller $\boldsymbol{u}$ in the form of \eqref{eqn: fb_ctrl_format}, the differential feedback controller can be found as $\delta\boldsymbol{u}=\boldsymbol{K}(\boldsymbol{x})\delta\boldsymbol{x}$, where $\boldsymbol{K}=\frac{\partial\boldsymbol{k}}{\partial\boldsymbol{x}}$. The resulting differential closed-loop system in intrinsic coordinates is
\begin{equation}\label{eqn: sys_closed_diff}
    \begin{aligned}
        \dot{\boldsymbol{v}} &= \frac{d}{dt}(\boldsymbol{P}_S^\top\delta\boldsymbol{x})=\dot{\boldsymbol{P}}_S^\top\boldsymbol{S}\boldsymbol{v} + \boldsymbol{P}_S^\top\delta\dot{\boldsymbol{x}} \\
        &=(\dot{\boldsymbol{P}}_S^\top+\boldsymbol{P}_S^\top\boldsymbol{A}+\boldsymbol{E}\boldsymbol{K})\boldsymbol{S}\boldsymbol{v} + \boldsymbol{E}_w\delta\boldsymbol{w} \\
        &= \boldsymbol{\mathcal{A}}\boldsymbol{v}+\boldsymbol{E}_w\delta\boldsymbol{w}, \\
        \delta\boldsymbol{z} &=(\boldsymbol{C}+\boldsymbol{D}\boldsymbol{K})\boldsymbol{S}\boldsymbol{v} = \boldsymbol{\mathcal{C}}\boldsymbol{v},
    \end{aligned}
\end{equation}
where $\dot{\boldsymbol{P}}_S^\top=\partial_{\dot{\boldsymbol{x}}}\boldsymbol{P}_S^\top$ with $\dot{\boldsymbol{x}}$ given by \eqref{eqn: sys_ctrl_affine_dist}, $\boldsymbol{\mathcal{A}}=(\dot{\boldsymbol{P}}_S^\top+\boldsymbol{P}_S^\top\boldsymbol{A}+\boldsymbol{E}\boldsymbol{K})\boldsymbol{S}$ and $\boldsymbol{\mathcal{C}}=(\boldsymbol{C}+\boldsymbol{D}\boldsymbol{K})\boldsymbol{S}$.

\begin{proposition}\label{prop: L-infty}
    A closed-loop system satisfying Assumption \ref{assum: sys_lie} admits a universal $\mathcal{L}_\infty$ gain bound of $\alpha>0$, if there exist a uniformly bounded symmetric metric $\boldsymbol{\mathcal M}(\boldsymbol{x})$, a feedback controller $\boldsymbol{u}$ of the form of \eqref{eqn: fb_ctrl_format} and constants $\lambda,\mu>0$, such that for all $\boldsymbol{x},\boldsymbol{w}$,
    \begin{equation}\label{eqn: LMI_R1}
\boldsymbol{R}_1=\begin{bmatrix}
    \dot{\boldsymbol{\mathcal M}} + \big\langle \boldsymbol{\mathcal M}\boldsymbol{\mathcal{A}} \big\rangle + 2\lambda \boldsymbol{\mathcal M} & \boldsymbol{\mathcal{M}}\boldsymbol{E}_w \\
    \boldsymbol{E}_w^\top\boldsymbol{\mathcal{M}} & -\mu \boldsymbol{I}_p
\end{bmatrix}
\preceq \boldsymbol{0},
\end{equation}
\begin{equation}\label{eqn: LMI_R2}
\boldsymbol{R}_2=\begin{bmatrix}
    2\lambda \boldsymbol{\mathcal M} - \alpha^{-1}\boldsymbol{\mathcal{C}}^\top\boldsymbol{\mathcal{C}}  & \boldsymbol{0} \\
    \boldsymbol{0} & (\alpha-\mu) \boldsymbol{I}_p
\end{bmatrix}
\succeq \boldsymbol{0},
\end{equation}
where $\dot{\boldsymbol{\mathcal M}} = \partial_{\dot{\boldsymbol{x}}}\boldsymbol{\mathcal M}$ with $\dot{\boldsymbol{x}}$ given by \eqref{eqn: sys_ctrl_affine_dist}. 
\end{proposition}

\begin{proof}
     The following construction is inspired by \cite{Zhao2022}, with modifications tailored to our setting. A smoothly parametrized path $\boldsymbol{c}(t)\in\Gamma_{\mathcal{X}}(\boldsymbol{x}(t),\boldsymbol{x}^*(t))$ is first defined, where $\boldsymbol{c}(t,0)=\boldsymbol{x}(t)$ and $\boldsymbol{c}(t,1)=\boldsymbol{x}^*(t)$. At each fixed time $t = t_i\in[0, \infty)$, we choose $\boldsymbol{c}(t)$ to be $\boldsymbol{\gamma}_{\boldsymbol{\mathcal M}}(t)$, the energy minimizing geodesic under metric $\boldsymbol{\mathcal{M}}$, which is defined as
     \begin{equation}
     \boldsymbol{\gamma}_{\boldsymbol{\mathcal M}}(t)=\arg_{\boldsymbol{c}\in\Gamma_{\mathcal{X}}(\boldsymbol{x},\boldsymbol{x}^*)}\min{\mathcal{E}(\boldsymbol{c}(t))}.
     \end{equation}
     Note that $\mathcal{E}(\boldsymbol{x}(t),\boldsymbol{x}^*(t))=\mathcal{E}(\boldsymbol{\gamma}_{\boldsymbol{\mathcal{M}}}(t))$. Based on the system in \eqref{eqn: sys_ctrl_affine_dist} and the feedback controller in \eqref{eqn: fb_ctrl_format}, we define the corresponding parameterized input, disturbance, and output paths for $s\in[0,1]$ as     
    \begin{equation}\label{eqn: para_paths}
        \begin{aligned}
        \boldsymbol{c}(t,s) &= \boldsymbol{\gamma}_{\boldsymbol{\mathcal M}}(t,s),\\
        \boldsymbol{\nu}(t,s) &= \boldsymbol{k}(\boldsymbol{c}(t,s),\boldsymbol{x}^*(t))+\boldsymbol{u}^*(t), \\
        \boldsymbol{\omega}(t,s) &= (1-s)\boldsymbol{w}^*(t)+s\boldsymbol{w}(t), \\
        \boldsymbol{\zeta}(t,s) &= \boldsymbol{g}(\boldsymbol{c}(t,s), \boldsymbol{\nu}(t,s)).\\
        \end{aligned}
    \end{equation}
    For any fixed time $t=t_i$, differentiate the smooth paths with respect to $s$, denoting $\frac{\partial}{\partial s}$ with subscript $s$, to get 
    \begin{equation}\label{eqn: para_paths_diff}
        \begin{aligned}
        \boldsymbol{c}_s (t,s) &= \tfrac{\partial\boldsymbol{\gamma}_{\boldsymbol{\mathcal{M}}}}{\partial s} (t,s),\\
        \boldsymbol{\nu}_s (t,s) &= \boldsymbol{K}(\boldsymbol{c}(t,s))\boldsymbol{c}_s (t,s),\\
        \boldsymbol{\omega}_s (t,s) &= \boldsymbol{w}(t) - \boldsymbol{w}^*(t),\\
        \boldsymbol{\zeta}_s (t,s) &= \boldsymbol{C}(\boldsymbol{c}(t,s), \boldsymbol{\nu}(t,s))\boldsymbol{c}_s (t,s)\\ &\,\,\,\,\,+\boldsymbol{D}(\boldsymbol{c}(t,s), \boldsymbol{\nu}(t,s))\boldsymbol{\nu}_s (t,s).
        \end{aligned}
    \end{equation}    
    where $\tfrac{\partial \boldsymbol{\gamma}_{\boldsymbol{\mathcal{M}}}}{\partial s} (t,s)=\boldsymbol{S}(\boldsymbol{\gamma}_{\boldsymbol{\mathcal{M}}}(t,s))\boldsymbol{\vartheta}(t,s) \in T_{\boldsymbol{\gamma}_{\boldsymbol{\mathcal{M}}}(t,s)}\mathcal{X}$, and $\boldsymbol{\vartheta}(t,s)$ represents parameterized tangent vectors in intrinsic coordinates according to the definition in \eqref{eqn: extrinsic_to_intrinsic}. On each short interval $[t_i,t_i+\epsilon)$, we fix $(\boldsymbol{\nu},\boldsymbol{\omega})$ at time $t_i$ and propagate $\boldsymbol{c}(t,s)$ under \eqref{eqn: sys_ctrl_affine_dist} for each $s$ with initial conditions $\boldsymbol{c}(t_i,s)$, yielding the differential dynamics similar to \eqref{eqn: sys_closed_diff} with $\boldsymbol{v}=\boldsymbol{\vartheta}$, $\delta\boldsymbol{w}=\boldsymbol{\omega}_s$, $\delta\boldsymbol{z}=\boldsymbol{\zeta}_s$
    \begin{equation}\label{eqn: para_diff_closed_loop}
        \begin{aligned}
            &\dot{\boldsymbol{\vartheta}} = \boldsymbol{\mathcal{A}}\boldsymbol{\vartheta} + \boldsymbol{E}_w \boldsymbol{\omega}_s, \\
            &\boldsymbol{\zeta}_s = \boldsymbol{\mathcal{C}}\boldsymbol{\vartheta}.
        \end{aligned}
    \end{equation}  

    Now, left-multiply and right-multiply the vector $[\boldsymbol{\vartheta}^\top \boldsymbol{\omega}_s^\top]$ and its transpose respectively to the LMI in \eqref{eqn: LMI_R1} gives us the inequality $\boldsymbol{\vartheta}^\top\big(\dot{\boldsymbol{\mathcal M}} + \big\langle \boldsymbol{\mathcal M}\boldsymbol{\mathcal{A}} \big\rangle \big)\boldsymbol{\vartheta} + 2\boldsymbol{\vartheta}^\top\boldsymbol{\mathcal{M}}\boldsymbol{E}_w\boldsymbol{\omega}_s + 2\lambda\boldsymbol{\vartheta}^\top\boldsymbol{\mathcal M}\boldsymbol{\vartheta} - \mu\boldsymbol{\omega}_s^\top\boldsymbol{\omega}_s\leq 0$. Combining with \eqref{eqn: para_diff_closed_loop} gives us
    \begin{equation}\label{eqn: lyap_inequality_para}
        \frac{d}{dt}(\boldsymbol{\vartheta}^\top\boldsymbol{\mathcal{M}}\boldsymbol{\vartheta}) \leq -2\lambda\boldsymbol{\vartheta}^\top\boldsymbol{\mathcal{M}}\boldsymbol{\vartheta} + \mu\boldsymbol{\omega}_s^\top\boldsymbol{\omega}_s.
    \end{equation}
    Integrating \eqref{eqn: lyap_inequality_para} over $s\in[0,1]$ and interchanging integration and differentiation results in  
    \begin{equation}\label{eqn: energy_path_inequality}
        \frac{d}{dt}\mathcal{E}(\boldsymbol{c}(t))\leq -2\lambda \mathcal{E}(\boldsymbol{c}(t)) + \mu\|\boldsymbol{w}(t) - \boldsymbol{w}^*(t)\|^2
    \end{equation}
    for $t\in[t_i,t_i+\epsilon)$. For sufficiently small $\epsilon$, \eqref{eqn: energy_path_inequality} indicates 
    \begin{equation}
        \begin{aligned}
            \mathcal{E}(\boldsymbol{c}(t)) &\leq \mathcal{E}(\boldsymbol{c}(t_i)) e^{-2\lambda (t - t_i)} \\
            &+ \mu \int_{t_i}^{t} e^{-2\lambda (t - \tau)} \| \boldsymbol{w}(\tau) - \boldsymbol{w}^*(\tau) \|^2 d\tau
        \end{aligned}
    \end{equation}
    by the comparison lemma. Given that $\mathcal{E}(\boldsymbol{\gamma}_{\boldsymbol{\mathcal{M}}}(t))\leq \mathcal{E}(\boldsymbol{c}(t))$ by definition of $\boldsymbol{\gamma}_{\boldsymbol{\mathcal{M}}}$ and $\mathcal{E}(\boldsymbol{c}(t_i)) = \mathcal{E}(\boldsymbol{\gamma}_{\boldsymbol{\mathcal{M}}}(t_i))$ by construction, we can conclude $\mathcal{E}(\boldsymbol{\gamma}_{\boldsymbol{\mathcal{M}}}(t)) \leq \mathcal{E}(\boldsymbol{\gamma}_{\boldsymbol{\mathcal{M}}}(t_i)) e^{-2\lambda (t - t_i)} + \mu \int_{t_i}^{t} e^{-2\lambda (t - \tau)} \| \boldsymbol{w}(\tau) - \boldsymbol{w}^*(\tau) \|^2 d\tau$. Since $t_i$ is arbitrary, taking $\epsilon\to 0$ gives us
    \begin{equation}\label{eqn: energy_geodesic_inequality}
        \frac{d}{dt}\mathcal{E}(\boldsymbol{\gamma}_{\boldsymbol{\mathcal{M}}}(t))\leq -2\lambda \mathcal{E}(\boldsymbol{\gamma}_{\boldsymbol{\mathcal{M}}}(t)) + \mu\|\boldsymbol{w}(t) - \boldsymbol{w}^*(t)\|^2
    \end{equation}    
    for all $t\in [0,\infty)$. Integrating \eqref{eqn: energy_geodesic_inequality} from $0$ to $t$ gives
    \begin{equation}\label{eqn: energy_geodesic_bound}
        \mathcal{E}(\boldsymbol{\gamma}_{\boldsymbol{\mathcal{M}}}(t)) \leq \mathcal{E}(\boldsymbol{\gamma}_{\boldsymbol{\mathcal{M}}}(0))e^{-2\lambda t} + \frac{\mu}{2\lambda}\| \boldsymbol{w} - \boldsymbol{w}^*\|^2_{\mathcal{L}_\infty^{[0,t]}}.
    \end{equation}

    We then left- and right-multiply the vector $[\boldsymbol{\vartheta}^\top \boldsymbol{\omega}_s^\top]$ and its transpose respectively to the LMI in \eqref{eqn: LMI_R2}, which by rearranging gives us the inequality $\alpha^{-1}\boldsymbol{\zeta}_s^\top\boldsymbol{\zeta}_s\leq2\lambda\boldsymbol{\vartheta}^\top\boldsymbol{\mathcal{M}}\boldsymbol{\vartheta} + (\alpha-\mu)\boldsymbol{\omega}_s^\top\boldsymbol{\omega}_s\leq 0$. Integrating this inequality along arbitrary path $\boldsymbol{c}(t)$ gives 
    \begin{equation}\label{eqn: R2_ineq_diff}
        \frac{1}{\alpha}\int^1_0\|\boldsymbol{\zeta}_s(t,s)\|^2ds\leq 2\lambda \mathcal{E}(\boldsymbol{c}(t))+(\alpha-\mu)\|\boldsymbol{w}(t)-\boldsymbol{w}^*(t)\|^2.
    \end{equation}
    Applying the Cauchy-Schwarz inequality, we have $\int^1_0\|\boldsymbol{\zeta}_s(t,s)\|^2ds\geq\|\int^1_0\boldsymbol{\zeta}_s(t,s)ds \|^2=\|\boldsymbol{z}(t)-\boldsymbol{z}^*(t)\|^2$. Choosing $\boldsymbol{c}(t)=\boldsymbol{\gamma}_{\boldsymbol{\mathcal{M}}}(t)$ and the $\mathcal{L}_\infty$ norm for $\boldsymbol{w}$, we get
    \begin{equation}\label{eqn: energy_output_bound}
        \|\boldsymbol{z}(t)-\boldsymbol{z}^*(t)\|^2\leq2\alpha\lambda \mathcal{E}(\boldsymbol{\gamma}_{\boldsymbol{\mathcal{M}}}(t)) + \alpha(\alpha-\mu)\|\boldsymbol{w}-\boldsymbol{w}^*\|^2_{\mathcal{L}_{\infty}^{[0,t]}}.
    \end{equation}
    Combining \eqref{eqn: energy_geodesic_bound} and \eqref{eqn: energy_output_bound} yields $\|\boldsymbol{z}(t)-\boldsymbol{z}^*(t)\|^2\leq \alpha^2 \|\boldsymbol{w}-\boldsymbol{w}^*\|^2_{\mathcal{L}_{\infty}^{[0,t]}} + 2\alpha\lambda \mathcal{E}(\boldsymbol{\gamma}_{\boldsymbol{\mathcal{M}}}(0))e^{-2\lambda t}$ for any $t$. Therefore, for any $T>0$, 
    \begin{equation}
        \|\boldsymbol{z}-\boldsymbol{z}^*\|^2_{\mathcal{L}_{\infty}^{[0,T]}}\leq \alpha^2 \|\boldsymbol{w}-\boldsymbol{w}^*\|^2_{\mathcal{L}_{\infty}^{[0,T]}} + \beta(\boldsymbol{x}(0), \boldsymbol{x}^*(0)),
    \end{equation}
    where $\beta(\boldsymbol{x}, \boldsymbol{x}^*)=2\alpha\lambda \mathcal{E}(\boldsymbol{x},\boldsymbol{x}^*)$.
\end{proof}

\textit{Remark:} A closed-loop system satisfying Assumption \ref{assum: sys_lie} and the LMI's in (\ref{eqn: LMI_R1}), (\ref{eqn: LMI_R2}) implies that the differential system (\ref{eqn: sys_differential_dist}) has a differential $\mathcal L_\infty$ gain bound of $\alpha>0$ according to the definition in \cite{Zhao2022}, such that for any $t\in[0,T]$, 
\begin{equation}
        \|\delta\boldsymbol{z}(t)\|^2\leq \alpha^2 \|\delta\boldsymbol{w}\|^2_{\mathcal L_{\infty}^{[0, T]}} + 2\lambda\alpha V(\boldsymbol{x}(0),\delta\boldsymbol{x}(0)).
\end{equation}

\begin{proof}
    First, left-multiply the vector $
        [\boldsymbol{v}^\top, \delta\boldsymbol{w}^\top]
    $ and right-multiply its transpose to the LMI (\ref{eqn: LMI_R2}). After rearranging, we get the inequality
    \begin{equation}
    \delta\boldsymbol{z}^\top\delta\boldsymbol{z}\leq\alpha(2\lambda V+ (\alpha-\mu)\delta\boldsymbol{w}^\top\delta\boldsymbol{w}).
    \end{equation}
    Moreover, performing the same left and right multiplication to the LMI (\ref{eqn: LMI_R1}), it gives us the inequality $\boldsymbol{v}^\top\big(\dot{\boldsymbol{\mathcal M}} + \big\langle \boldsymbol{\mathcal M}\boldsymbol{\mathcal{A}} \big\rangle \big)\boldsymbol{v} + 2\boldsymbol{v}^\top\boldsymbol{\mathcal{M}}\boldsymbol{E}_w\delta\boldsymbol{w} + 2\lambda\boldsymbol{v}^\top\boldsymbol{\mathcal M}\boldsymbol{v} - \mu\delta\boldsymbol{w}^\top\delta\boldsymbol{w}\leq 0$. Defining a function $V(\boldsymbol{x}(t), \boldsymbol{v}(t))$ as in (\ref{eqn: intrinsic_diff_CLF}), we can obtain the following inequality
    \begin{equation}\label{eqn: lyap_inequality}
        \dot{V}(\boldsymbol{x}(t), \delta\boldsymbol{x}(t)) \leq -2\lambda V(\boldsymbol{x}(t), \delta\boldsymbol{x}(t)) + \mu\delta\boldsymbol{w}^\top\delta\boldsymbol{w}
    \end{equation}
    This implies $\dot{V}\leq 0$ when $V\geq  \frac{\mu}{2\lambda}\|\delta\boldsymbol{w}\|^2$, and therefore $V(\boldsymbol{x}(t), \delta\boldsymbol{x}(t)) \leq \max\{V(\boldsymbol{x}(0),\delta\boldsymbol{x}(0)), \frac{\mu}{2\lambda}\|\delta\boldsymbol{w}\|^2_{\mathcal L_{\infty}^{[0, T]}}\}\leq V(\boldsymbol{x}(0),\delta\boldsymbol{x}(0))+ \frac{\mu}{2\lambda}\|\delta\boldsymbol{w}\|^2_{\mathcal L_{\infty}^{[0, T]}}$ for any $t\in[0, T]$. Plugging this inequality into (\ref{eqn: lyap_inequality}), we get 
    \begin{equation}
        \|\delta\boldsymbol{z}(t)\|^2\leq \alpha^2 \|\delta\boldsymbol{w}\|^2_{\mathcal L_{\infty}^{[0, T]}} + 2\lambda\alpha V(\boldsymbol{x}(0),\delta\boldsymbol{x}(0))
    \end{equation}
    for any $t\in[0,T]$.
\end{proof}

The metric $\boldsymbol{\mathcal M}$ that satisfies the LMIs in \eqref{eqn: LMI_R1} and \eqref{eqn: LMI_R2} is termed an RCCM and the associated feedback controller will guarantee input-output stability for the closed-loop system.

\begin{corollary}  
    The output vector $\boldsymbol{z}$ lives in $\mathbb R^l$ in our derivations. A tube-like forward invariant set can be characterized in the output Euclidean space such that after transient effects, 
    \begin{equation}\label{eqn: tube_euclidean}
        \boldsymbol{z}\in\Omega(\boldsymbol{z}^*)=\{\boldsymbol{y}\in\mathbb R^l\mid\|\boldsymbol{y}-\boldsymbol{z}^*\|\leq\alpha\overline{w}\},
    \end{equation}
    where $\overline{w}=\|\boldsymbol{w}-\boldsymbol{w}^*\|_{\mathcal{L}_\infty}$. In the case where the output space lives in a Lie group and $\boldsymbol{z}\in\mathcal{Z}\subseteq\mathcal{X}\times\mathbb R^m$, output deviations can be measured using Riemannian metrics on $\mathcal{Z}$. Similar to \eqref{eqn: extrinsic_to_intrinsic}, we can decompose tangent vectors on $\mathcal{Z}$ such that $\delta\boldsymbol{z}=\boldsymbol{S}_z(\boldsymbol{z})\boldsymbol{v}_z\in T_{\boldsymbol{z}}\mathcal{Z}$. The uniform boundedness of $\boldsymbol{S}_z^\top\boldsymbol{S}_z$ still holds, as $\boldsymbol{S}^\top\boldsymbol{S}$ is uniformly bounded from Assumption \ref{assum: sys_lie} and basis of tangent spaces in $\mathbb R^m$ is represented by $\boldsymbol{I}_m$.
    
    A natural choice for Riemannian metrics on Lie groups will be inner products of the left-trivialized coordinates $\boldsymbol{v}_z$. We choose the metric $g(\boldsymbol{v}_{z,1},\boldsymbol{v}_{z,2})=\boldsymbol{v}_{z,1}^\top\boldsymbol{Q}\boldsymbol{v}_{z,2}=\delta\boldsymbol{z}_1^\top \boldsymbol{P}_{S_z}(\boldsymbol{z})\boldsymbol{Q}\boldsymbol{P}_{S_z}^\top(\boldsymbol{z})\delta\boldsymbol{z}_2$, where $\boldsymbol{P}_{S_z}=\boldsymbol{S}_z(\boldsymbol{S}_z^\top\boldsymbol{S}_z)^{-1}$. $\boldsymbol{Q}$ is a weighting matrix for each left-trivialized coordinate, so that the input-output gain can be tuned for each channel. Therefore, a generalized version of LMI \eqref{eqn: LMI_R2} is proposed as 
    \begin{equation}\label{eqn: LMI_R2_prime}
    \boldsymbol{R}_2'=\begin{bmatrix}
        2\lambda \boldsymbol{\mathcal{M}}-\alpha^{-1}\boldsymbol{\mathcal{C}}^\top\boldsymbol{P}_{S_z}\boldsymbol{Q}\boldsymbol{P}_{S_z}^\top\boldsymbol{\mathcal{C}}  & \boldsymbol{0} \\
        \boldsymbol{0} & (\alpha-\mu) \boldsymbol{I}_p
    \end{bmatrix}
    \succeq \boldsymbol{0}.
    \end{equation}
    This guarantees the Riemannian length between $\boldsymbol{z}$ and $\boldsymbol{z}^*$ has a universal $\mathcal{L}_\infty$ gain bound of $\alpha$, and a similar invariant set can be characterized on the output manifold such that 
    \begin{equation}
        \boldsymbol{z}\in\Omega_{\mathcal{Z}}(\boldsymbol{z}^*)=\{\boldsymbol{y}\in \mathcal{Z}\mid d_g(\boldsymbol{y},\boldsymbol{z}^*)\leq\alpha\overline{w}\},
    \end{equation}
    where $d_g(\boldsymbol{z},\boldsymbol{z}^*)$ is the Riemannian distance between $\boldsymbol{z},\boldsymbol{z}^*$ under metric $g$.
\end{corollary}
\begin{proof}   
    Revisiting the proof of Proposition \ref{prop: L-infty}, we now have $\boldsymbol{\zeta}_s(t,s)=\boldsymbol{S}_z(\boldsymbol{\zeta}(t,s))\boldsymbol{\vartheta}_z$ for the derivative of the parametrized output in (\ref{eqn: para_paths_diff}). Following similar procedures of left- and right-multiplication of vector $[\boldsymbol{\vartheta}^\top \boldsymbol{\omega}_s^\top]$ and its transpose respectively to the LMI (\ref{eqn: LMI_R2_prime}), a generalized version of inequality (\ref{eqn: R2_ineq_diff}) can be obtained as
    \begin{equation}\label{eqn: R2_ineq_diff_new}
        \frac{1}{\alpha}\mathcal{E}_{\mathcal{Z}}(\boldsymbol{\zeta}(t))\leq 2\lambda \mathcal{E}(\boldsymbol{c}(t))+(\alpha-\mu)\|\boldsymbol{w}(t)-\boldsymbol{w}^*(t)\|^2,
    \end{equation}
    where $\mathcal{E}_{\mathcal{Z}}(\boldsymbol{\zeta}(t)) = \int^1_0g(\boldsymbol{\vartheta}_z(t,s),\boldsymbol{\vartheta}_z(t,s))ds$. Defining the set $\Gamma_{\mathcal{Z}}(\boldsymbol{z},\boldsymbol{z}^*) = \{\boldsymbol{\zeta}(s)=\boldsymbol{g}\big(\boldsymbol{c}(s),\boldsymbol{k}(\boldsymbol{c}(s),\boldsymbol{x}^*+\boldsymbol{u}^*\big) \mid \boldsymbol{c} \in \Gamma_{\mathcal{X}}(\boldsymbol{x},\boldsymbol{x}^*)\}$, we have, by Hopf-Rinow theorem, for all $\boldsymbol{\zeta}(t)\in\Gamma_{\mathcal{Z}}(\boldsymbol{z}(t),\boldsymbol{z}^*(t))$
    \begin{equation}
        d_g^2(\boldsymbol{z}(t),\boldsymbol{z}^*(t))\leq \mathcal{E}_{\mathcal{Z}}(\boldsymbol{\zeta}(t)) 
    \end{equation}
    where $d_g(\boldsymbol{z}(t),\boldsymbol{z}^*(t))$ is the Riemannian distance between $\boldsymbol{z},\boldsymbol{z}^*$ under metric $g$. We therefore have 
    \begin{equation}\label{eqn: R2_ineq_new}
        \frac{1}{\alpha}d_g^2(\boldsymbol{z}(t),\boldsymbol{z}^*(t)) \leq 2\lambda \mathcal{E}(\boldsymbol{c}(t))+(\alpha-\mu)\|\boldsymbol{w}(t)-\boldsymbol{w}^*(t)\|^2.
    \end{equation}
    It then follows easily that
    \begin{equation}
        d_g^2(\boldsymbol{z},\boldsymbol{z}^*)_{\mathcal{L}_{\infty}^{[0,T]}}\leq \alpha^2 \|\boldsymbol{w}-\boldsymbol{w}^*\|^2_{\mathcal{L}_{\infty}^{[0,T]}} + \beta(\boldsymbol{x}(0), \boldsymbol{x}^*(0)).
    \end{equation}
    We can therefore characterize a tube-like bound on the output manifold similar to (\ref{eqn: tube_euclidean}) such that
    \begin{equation}
        \boldsymbol{z}\in\Omega_{\mathcal{Z}}(\boldsymbol{z}^*)=\{\boldsymbol{y}\in \mathcal{Z}\mid d_g(\boldsymbol{z},\boldsymbol{z}^*)\leq\alpha\overline{w}\}.
    \end{equation}
\end{proof}

\subsection{Conditions for the existence of RCCM controllers}

From \eqref{eqn: LMI_R1}, $\dot{\boldsymbol{\mathcal M}} + \big\langle \boldsymbol{\mathcal M}\boldsymbol{\mathcal{A}} \big\rangle + 2\lambda \boldsymbol{\mathcal M}$ depends on $\boldsymbol{u}$ and $\boldsymbol{w}$, which are not known prior to implementation. Adopting from \cite{Manchester2017}, extra conditions are proposed to remove these dependencies. Through musical isomorphism $\boldsymbol{\mathcal{W}}=\boldsymbol{\mathcal{M}}^{-1}$ and coordinate change $\boldsymbol{\eta}=\boldsymbol{\mathcal{M}}\boldsymbol{v}\in\mathbb{R}^q$, a dual contraction condition that preserves contraction guarantees is obtained. For all $\boldsymbol{\eta}\neq\boldsymbol{0}$, $\boldsymbol{\eta}^\top\boldsymbol{E}=\boldsymbol{0}$, we have $\boldsymbol{\eta}^\top(-\dot{\boldsymbol{\mathcal{W}}}+\big\langle (\dot{\boldsymbol{P}}_S^\top+\boldsymbol{P}_S^\top\boldsymbol{A})\boldsymbol{S}\boldsymbol{\mathcal W}\big\rangle + 2\lambda \boldsymbol{\mathcal W})\boldsymbol{\eta} < 0$, which can be rewritten as
\begin{equation}\label{eqn: dual_contraction_condition}
    \boldsymbol{E}_\perp^\top(-\dot{\boldsymbol{\mathcal{W}}}+\big\langle (\dot{\boldsymbol{P}}_S^\top+\boldsymbol{P}_S^\top\boldsymbol{A})\boldsymbol{S}\boldsymbol{\mathcal W}\big\rangle + 2\lambda \boldsymbol{\mathcal W})\boldsymbol{E}_\perp\prec \boldsymbol{0},
\end{equation}
where $\boldsymbol{E}_\perp$ is a full rank matrix function such that $\boldsymbol{E}_\perp^\top\boldsymbol{E}=\boldsymbol{0}$. With $\boldsymbol{S}_{\boldsymbol{f}}=(\partial_{\boldsymbol{f}}\boldsymbol{P}_S^\top+\boldsymbol{P}_S^\top\frac{\partial \boldsymbol{f}}{\partial\boldsymbol{x}})\boldsymbol{S}$, $\boldsymbol{S}_{\boldsymbol{b}_i}=(\partial_{\boldsymbol{b}_i}\boldsymbol{P}_S^\top+\boldsymbol{P}_S^\top\frac{\partial \boldsymbol{b}_i}{\partial\boldsymbol{x}})\boldsymbol{S}$ and $\boldsymbol{S}_{\boldsymbol{b}_{w,j}}=(\partial_{\boldsymbol{b}_{w,j}}\boldsymbol{P}_S^\top+\boldsymbol{P}_S^\top\frac{\partial \boldsymbol{b}_{w,j}}{\partial\boldsymbol{x}})\boldsymbol{S}$, the following conditions are proposed. For all $\boldsymbol{x}\in\mathcal{X}$, $i\in \mathbb Z_m$ and $j \in \mathbb Z_p$,
    \begin{equation}\label{eqn: LMI_C1}
    \boldsymbol{C}_1=\boldsymbol{E}_{\perp}^\top\left( -\partial_{\boldsymbol{f}}\boldsymbol{\mathcal W} + \langle\boldsymbol{S}_{\boldsymbol{f}}\boldsymbol{\mathcal W}\rangle + 2\lambda\boldsymbol{\mathcal W}\right)\boldsymbol{E}_{\perp} \prec \boldsymbol{0},
    \end{equation}
    \begin{equation}\label{eqn: cond_C2}
    \boldsymbol{C}_{2,i}=\boldsymbol{E}_{\perp}^\top\big( -\partial_{\boldsymbol{b}_i}\boldsymbol{\mathcal W} + \langle\boldsymbol{S}_{\boldsymbol{b}_i}\boldsymbol{\mathcal W}\rangle \big)\boldsymbol{E}_{\perp} = \boldsymbol{0},
    \end{equation}
    \begin{equation}\label{eqn: cond_C2w}
    \boldsymbol{C}_{3,j}=\boldsymbol{E}_{\perp}^\top\big( -\partial_{\boldsymbol{b}_{w,j}}\boldsymbol{\mathcal W} + \langle\boldsymbol{S}_{\boldsymbol{b}_{w,j}}\boldsymbol{\mathcal W}\rangle \big)\boldsymbol{E}_{\perp} = \boldsymbol{0}.
    \end{equation}
\begin{proposition}
   Conditions \eqref{eqn: cond_C2} and \eqref{eqn: cond_C2w} remove the dependence on $\boldsymbol{u}$ and $\boldsymbol{w}$, respectively, in $\boldsymbol{R}_1$ of \eqref{eqn: LMI_R1}. Together with the LMI in \eqref{eqn: LMI_C1}, these conditions guarantee the existence of a feedback controller of the form \eqref{eqn: fb_ctrl_format} with differential control law $\delta\boldsymbol{u}=\boldsymbol{K}\delta\boldsymbol{x}$, and some appropriately chosen constants $\alpha,\mu$, such that \eqref{eqn: LMI_R1} and \eqref{eqn: LMI_R2} hold.   
\end{proposition}

\begin{proof}
    The dual contraction condition in \eqref{eqn: dual_contraction_condition} can be decomposed into
    \begin{equation}
        \boldsymbol{C}_1 + \sum_{i=1}^m\boldsymbol{C}_{2,i}\cdot u_i + \sum_{j=1}^p\boldsymbol{C}_{3,j}\cdot w_j \prec \boldsymbol{0},
    \end{equation}
    such that conditions \eqref{eqn: LMI_C1}, \eqref{eqn: cond_C2} and \eqref{eqn: cond_C2w} guarantees the LMI in \eqref{eqn: dual_contraction_condition} holds. Note that these three conditions are independent of $\boldsymbol{u}$ and $\boldsymbol{w}$. Formally, \eqref{eqn: LMI_C1} makes sure the uncontrolled system is contracting in directions orthogonal to the span of the control inputs, while \eqref{eqn: cond_C2} and \eqref{eqn: cond_C2w} enforce $\boldsymbol{b}_i$ and $\boldsymbol{b}_{w,j}$ to be Killing vector fields. As shown in Proposition 2 of \cite{Manchester2017}, satisfying the above conditions is equivalent to the existence of a differential feedback controller in the form $\delta\boldsymbol{u}=\boldsymbol{K}(\boldsymbol{x})\delta\boldsymbol{x}$ that satisfies $\dot{\boldsymbol{\mathcal M}} + \big\langle \boldsymbol{\mathcal M}\boldsymbol{\mathcal{A}} \big\rangle + 2\lambda \boldsymbol{\mathcal M}\prec \boldsymbol{0}$. Consequently, $\delta\boldsymbol{u}$ is path integrable and the control can be found in the general form \eqref{eqn: fb_ctrl_format}. Finally, if the upper-left block of $\boldsymbol{R}_1$ is negative definite, an appropriate choice of $\mu$ (e.g., via line search) ensures that $\boldsymbol{R}_1\preceq\boldsymbol{0}$. As $\boldsymbol{\mathcal{M}}$ is uniformly bounded, $\boldsymbol{R}_2\succeq\boldsymbol{0}$ is satisfied with a sufficiently large $\alpha$.
\end{proof}

\subsection{Offline learning of RCCMs and feedback controller}

To implement the theoretical results above, the feedback controller \eqref{eqn: fb_ctrl_format} and the dual metric $\boldsymbol{\mathcal{W}}$ are formulated as neural networks and learned simultaneously. The state, control and disturbance variables are limited to their compact subspaces to make learning feasible. The neural feedback controller is modelled as $\boldsymbol{u}=\boldsymbol{k}_{nn}(\boldsymbol{x},\boldsymbol{x}^*)+\boldsymbol{u}^*$, where $\boldsymbol{k}_{nn}(\boldsymbol{x},\boldsymbol{x}^*)=\boldsymbol{K}_1(\boldsymbol{x},\boldsymbol{x}^*;\theta_{k_1})\cdot \tanh{\big(\boldsymbol{K}_2(\boldsymbol{x},\boldsymbol{x}^*;\theta_{k_2})\cdot\boldsymbol{\varepsilon}(\boldsymbol{x},\boldsymbol{x}^*)}\big)$. $\theta_{k_1}$, $\theta_{k_2}$ are neural network parameters for $\boldsymbol{K}_1\in\mathbb{R}^{m\times3n}$ and $\boldsymbol{K}_2\in\mathbb{R}^{3n\times q}$, and $\tanh{(\cdot)}$ is the hyperbolic tangent function. $\boldsymbol{\varepsilon}(\cdot,\cdot):\mathcal{X}\times\mathcal{X}\to\mathbb R^q$ is an error function for Lie groups represented in Euclidean space with $\boldsymbol{\varepsilon}(\boldsymbol{x},\boldsymbol{x})=\boldsymbol{0}$ $\forall \boldsymbol{x}\in\mathcal{X}$. This can be designed specifically for each Lie group, and further details are provided in Section \ref{sec: application}. This neural controller follows the structure of \eqref{eqn: fb_ctrl_format} by design, where $\boldsymbol{k}_{nn}(\boldsymbol{x},\boldsymbol{x})=\boldsymbol{0}\,\,\,\forall\boldsymbol{x}\in\mathcal{X}$. Moreover, the dual metric is modelled as $\boldsymbol{\mathcal{W}}(\boldsymbol{x})=\boldsymbol{\Theta}(\boldsymbol{x};\theta_{w})^\top\boldsymbol{\Theta}(\boldsymbol{x};\theta_{w})+\overline{\mathfrak{m}}^{-1}\boldsymbol{I}_q$, such that $\boldsymbol{\mathcal{W}}$ is lower-bounded by $\overline{\mathfrak{m}}^{-1}$. $\alpha$ and $\mu$ in \eqref{eqn: LMI_R1} and \eqref{eqn: LMI_R2} are parametrized by learnable variables $\theta_\alpha,\theta_\mu$, respectively, and trained jointly. Specifically, $\alpha=\ln(1+e^{\theta_\alpha})$ and $\mu=\ln(1+e^{\theta_\mu})$, which ensures $\alpha,\mu>0$. 

Let $\{(\boldsymbol{x}_k,\boldsymbol{x}^*_k,\boldsymbol{u}^*_k,\boldsymbol{w}_k)\}^N_{k=1}$ be $N$ training samples drawn uniformly from their respective compact subsets. The network parameters $\theta_{k_1}$, $\theta_{k_2}$, $\theta_{w}$, $\theta_\alpha$, $\theta_\mu$ are trained by minimizing the empirical loss $L=\frac{1}{N}\sum^N_{k=1}L_k$, where
\begin{equation}\label{eqn: training_loss_k}
\begin{aligned}
    &L_k = \operatorname{L_{PD}}(-\boldsymbol{R}_1) + \operatorname{L_{PD}}(\boldsymbol{R}_2) + 
    \operatorname{L_{PD}}(-\boldsymbol{C}_1) + \hspace{-0.08cm}\sum_{i=1}^m\|\boldsymbol{C}_{2,i}\|_F \\
    &+\sum_{j=1}^p\|\boldsymbol{C}_{3,j}\|_F + \operatorname{L_{PD}}(\underline{\mathfrak{m}}^{-1}\boldsymbol{I}_q - \boldsymbol{\mathcal{W}}) + \operatorname{ReLU}(\alpha-\underline{\alpha})
\end{aligned}
\end{equation}
is the loss for each training sample. Note that $\|\cdot\|_F$ is the Frobenius norm and $\operatorname{ReLU}(x)=\max(0,x)$. Given a matrix $\boldsymbol{A}\in \mathbb R^{n\times n}$, $\operatorname{L_{PD}}(\boldsymbol{A})\geq 0$ penalizes negative definiteness of $\boldsymbol{A}$. We uniformly sample $K$ unit vectors $\{\boldsymbol{p}_i\in\mathbb{R}^n \mid \|\boldsymbol{p}_i\|=1 \}^K_{i=1}$ and $\operatorname{L_{PD}}(\boldsymbol{A})=\frac{1}{K}\sum^K_{i=1}\max(0,-\boldsymbol{p}_i^\top \boldsymbol{A}\boldsymbol{p}_i)$. Loss terms based on LMIs \eqref{eqn: LMI_R1}, \eqref{eqn: LMI_R2} are added together as soft training loss. Auxiliary loss terms based on conditions \eqref{eqn: LMI_C1}, \eqref{eqn: cond_C2} and \eqref{eqn: cond_C2w} are introduced to ``guide'' the training. Consistent with the findings in \cite{Sun2021}, it becomes more challenging for the neural networks to converge without them. Moreover, $\operatorname{L_{PD}}(\underline{\mathfrak{m}}^{-1}\boldsymbol{I}_q - \boldsymbol{\mathcal{W}})$ penalizes the learned dual metric if it exceeds the upper bound $\underline{\mathfrak{m}}^{-1}$. This guarantees uniform boundedness of the RCCM as in \eqref{eqn: cond_metric_uni_bound}. Lastly, $\operatorname{ReLU}(\alpha-\underline{\alpha})$ is incorporated to the loss to minimize the $\mathcal{L}_\infty$ gain bound.

Choosing hyperparameters $(\underline{\mathfrak{m}}, \overline{\mathfrak{m}}, \underline{\alpha}, \lambda)$, the functions $\boldsymbol{\mathcal{W}}$ and $\boldsymbol{k}_{nn}$ are trained using gradient descent. Although the neural networks are trained on sampled data and residual learning errors may persist, theoretical guarantees can be established locally within the sample space \cite{Tsukamoto2021}[Theorem 6].

\subsection{Online implementation of neural feedback controller}\label{subsec: online_implementation}

One of the main advantages of this framework is that online computation of the control law is lightweight, since the computationally intensive task of learning the controller function is performed offline. This avoids integrating the differential control law along geodesics in common implementations \cite{Zhao2022, Tsukamoto2020}. In the closed-loop system, the states are fed back to the controller, and the motion planner provides a nominal trajectory $\boldsymbol{x}^*(t)$ and control input $\boldsymbol{u}^*(t)$ that jointly satisfy the system in \eqref{eqn: sys_ctrl_affine_dist}. If disturbance estimates are incorporated into motion planning, then $\boldsymbol{w}-\boldsymbol{w}^*$ is given by the disturbance estimation error, potentially reducing $\overline{w}$ and thus the size of the bounding tube of output $\boldsymbol{z}$. 

\section{Case Study: A Quadrotor Example}\label{sec: application}

The system dynamics of a rate-controlled quadrotor with mass $m$ is studied in this section. The state vector is defined as $\boldsymbol{x}=[\boldsymbol{p}_I^\top,\boldsymbol{v}_I^\top,\boldsymbol{r}_{IB}^\top]^\top$, where $\boldsymbol{p}_I\in\mathbb{R}^3$ is the inertial position, $\boldsymbol{v}_I\in\mathbb{R}^3$ is the inertial velocity, and $\boldsymbol{r}_{IB}=\operatorname{vec}(\boldsymbol{R}_{IB})\in\mathbb{R}^9$ is the vectorized rotation matrix from the body-fixed frame of the quadrotor to the inertial frame. The state space is a 9-dimensional Lie group $\mathcal{X}=\mathbb R^6\times SO(3)$ embedded in $\mathbb R^{15}$, in which a smooth $\boldsymbol{S}(\boldsymbol{x})$ representing the tangent bundle can be found. For $SO(3)$, we define $\boldsymbol{S}_r(\boldsymbol{x})\in\mathbb{R}^{9\times3}$ to represent the basis of $TSO(3)$
\begin{equation}
    \boldsymbol{S}_r(\boldsymbol{x}) = \big[\operatorname{vec}(\boldsymbol{R}_{IB}\boldsymbol{e}_1^\wedge)\,\operatorname{vec}(\boldsymbol{R}_{IB}\boldsymbol{e}_2^\wedge)\, \operatorname{vec}(\boldsymbol{R}_{IB}\boldsymbol{e}_3^\wedge)\big],
\end{equation}
where $\boldsymbol{S}_r^\top\boldsymbol{S}_r=2\boldsymbol{I}_3$. Hence, $\boldsymbol{S}(\boldsymbol{x})=\operatorname{diag}\big(
    \boldsymbol{I}_6,\, \boldsymbol{S}_r(\boldsymbol{x})\big)$ with $\boldsymbol{S}^\top\boldsymbol{S}$ uniformly bounded.
The control input is composed of the mass-normalized collective thrust $f_t/m\in\mathbb{R}$ and body angular rates $\boldsymbol{\omega}_B\in\mathbb{R}^3$, such that $\boldsymbol{u}=[f_t/m\,\,\,\boldsymbol{\omega}_B^\top]^\top$. The disturbance is composed of the mass-normalized disturbance force $\boldsymbol{f}_d/m\in\mathbb{R}^3$ and angular rate disturbance $\boldsymbol{\omega}_d\in\mathbb{R}^3$ such that $\boldsymbol{w}=[\boldsymbol{f}_d^\top/m,\,\,\boldsymbol{\omega}^\top_d]^\top$. The output is defined as $\boldsymbol{z}=[\boldsymbol{Q}\boldsymbol{p}_I^\top\,\,\,\boldsymbol{R}\boldsymbol{u}^\top]^\top\in\mathbb{R}^7$, where $\boldsymbol{Q}$ and $\boldsymbol{R}$ are weighting matrices. The system's dynamical equation is therefore written in the form of \eqref{eqn: sys_ctrl_affine_dist} as 
\begin{equation}\label{eqn: sys_quad}
    \dot{\boldsymbol{x}}=\begin{bmatrix}
        \boldsymbol{v}_I \\ \boldsymbol{g}_I \\ \boldsymbol{0}
    \end{bmatrix} + \begin{bmatrix}
        \boldsymbol{0} & \boldsymbol{0} \\
        \boldsymbol{R}_{IB}\boldsymbol{e}_3 & \boldsymbol{0} \\
        \boldsymbol{0} & \boldsymbol{S}_r \\
    \end{bmatrix}
    \begin{bmatrix}
        f_t/m \\ \boldsymbol{\omega}_B
    \end{bmatrix} + 
    \begin{bmatrix}
        \boldsymbol{0} & \boldsymbol{0}\\
        \boldsymbol{I}_3 & \boldsymbol{0}\\
        \boldsymbol{0} & \boldsymbol{S}_r\\
    \end{bmatrix}
    \begin{bmatrix}
        \boldsymbol{f}_d/m \\ \boldsymbol{\omega}_d
    \end{bmatrix},
\end{equation}
where $\boldsymbol{g}_I = [0, 0, -9,81]^\top$ is the gravitational acceleration vector in the inertial frame. One can easily verify that \eqref{eqn: sys_quad} satisfies Assumption \ref{assum: sys_lie}. The quadrotor mass is set as $1kg$, and the weighting matrices are chosen as $\boldsymbol{Q}=\boldsymbol{I}_3$ and $\boldsymbol{R}=\operatorname{Diag}(0.1, 0.5, 0.5, 0.5)$. Moreover, from inspection of \eqref{eqn: E_matrix_transformation} and \eqref{eqn: sys_quad}, $\boldsymbol{E}(\boldsymbol{x})\in\mathbb{R}^{9\times 4}$, $\boldsymbol{E}_w(\boldsymbol{x})\in\mathbb{R}^{9\times 6}$ can be found as
\begin{equation}
    \boldsymbol{E}(\boldsymbol{x})=\begin{bmatrix}
        \boldsymbol{0} & \boldsymbol{0} \\
        \boldsymbol{R}_{IB}\boldsymbol{e}_3 & \boldsymbol{0} \\
        \boldsymbol{0} & \boldsymbol{I}_3 \\
    \end{bmatrix},\, \boldsymbol{E}_w(\boldsymbol{x})=\begin{bmatrix}
        \boldsymbol{0} \\
        \boldsymbol{I}_6
    \end{bmatrix}. 
\end{equation}
The analytical solution for $\boldsymbol{E}_\perp(\boldsymbol{x})\in\mathbb{R}^{9\times5}$ can be found as
\begin{equation}
    \boldsymbol{E}_\perp(\boldsymbol{x}) = \begin{bmatrix}
        \boldsymbol{I}_3 & \boldsymbol{0} & \boldsymbol{0} \\
        \boldsymbol{0} & \boldsymbol{R}_{IB}\boldsymbol{e}_1 & \boldsymbol{R}_{IB}\boldsymbol{e}_2 \\
        \boldsymbol{0} & \boldsymbol{0} & \boldsymbol{0}
    \end{bmatrix},
\end{equation}
such that $\boldsymbol{E}_\perp^\top(\boldsymbol{x})\boldsymbol{E}(\boldsymbol{x})=\boldsymbol{0}$ for all $\boldsymbol{x}\in\mathcal{X}$. Lastly, the error function $\boldsymbol{\varepsilon}(\boldsymbol{x},\boldsymbol{x}^*)$ in $\boldsymbol{k}_{nn}(\boldsymbol{x},\boldsymbol{x}^*)$ is designed by adopting the geometric error from \cite{Lee2010}. We define an error in $SO(3)$ as $\boldsymbol{\varepsilon}_r=\frac{1}{2}(\boldsymbol{R}_{IB}^{*\top}\boldsymbol{R}_{IB}-\boldsymbol{R}_{IB}^\top\boldsymbol{R}_{IB}^*)^\vee\in\mathbb{R}^{3}$. 
Combining it with the Euclidean error $\boldsymbol{\varepsilon}_p=\boldsymbol{p}_I-\boldsymbol{p}_I^*$ and $\boldsymbol{\varepsilon}_v=\boldsymbol{v}_I-\boldsymbol{v}_I^*$ gives
\begin{equation}
    \boldsymbol{\varepsilon}(\boldsymbol{x},\boldsymbol{x}^*)=\begin{bmatrix}
        \boldsymbol{\varepsilon}_p^\top & \boldsymbol{\varepsilon}_v^\top & \boldsymbol{\varepsilon}_r^\top
    \end{bmatrix}^\top.
\end{equation}
\subsection{Training and simulation results}

Each training sample $(\boldsymbol{x}_k,\boldsymbol{x}^*_k,\boldsymbol{u}^*_k,\boldsymbol{w}_k)$ is drawn from compact subsets of $\mathcal{X}\times\mathcal{X}\times\mathbb{R}^m\times\mathbb{R}^p$. For Euclidean states, control and disturbances, samples are drawn within a range. To generate sample states in $SO(3)$, Euler angles are randomly sampled within $(-1,-1,-\pi)\leq(\phi,\theta,\psi)\leq(1,1,\pi)$ and converted into rotation matrices. Moreover, hyperparameters are selected as $\underline{\mathfrak m}=0.1$, $\overline{\mathfrak m}=10$, $\lambda=0.5$ and $\underline{\alpha}=0.7$. Each neural network is composed of 2 hidden layers with 128 neurons each and a $\tanh$ activation function. Training is performed for 30 epochs with $N=131072$ training samples, using the loss function in \eqref{eqn: training_loss_k}. Our GitHub repository in \cite{Lo2026} provides further implementation details. 

The RCCM controller is trained with \textit{Adam optimizer}, and an $\mathcal{L}_\infty$ gain bound of $\alpha=1.268$ is achieved. To verify the tracking performance of the controller, a spiral reference motion $(\boldsymbol{x}^*,\boldsymbol{u}^*)$ is generated using differential flatness-based motion planning following \cite{Faessler2017}. The flat variables are chosen as $\boldsymbol{p}_I^*(t)=[0.5t,3\cos(1.5t),3\sin(1.5t)]^\top$ and $\psi^*(t)=0$, where $\psi^*$ is the desired yaw angle. Moreover, artificial disturbances are injected into the system, which are modelled as $\boldsymbol{f}_d(t)=\big(0.8+0.2\sin(0.2\pi t)\big)[0.6, 0.7, 0.3]^\top$ and $\boldsymbol{\omega}_d(t)=\big(0.5+0.5\sin(2\pi t)\big)[0.1, 0.1, 0.2]^\top$. The $\mathcal{L}_\infty$ norm $\|\boldsymbol{w}(t)\|_{\mathcal{L}_\infty}$ is calculated to be $\overline{w}=1$, for $\boldsymbol{w}^*(t)=\boldsymbol{0}$.

As a comparative study, three extra test cases are implemented. First, a CCM controller is trained with the assumption of no noise. We therefore assumed $\boldsymbol{B}_w(\boldsymbol{x})=\boldsymbol{0}$ and the neural networks are trained with the loss function
\begin{equation}\label{eqn: training_loss_CCM}
\begin{aligned}
    &L = \frac{1}{N}\sum^N_{k=1}\bigg(\operatorname{L_{PD}}(-\boldsymbol{\mathcal{CCM}}) + \operatorname{L_{PD}}(-\boldsymbol{C}_1) \\ & + \sum_{i=1}^m\|\boldsymbol{C}_{2,i}\|_F + \operatorname{L_{PD}}(\underline{\mathfrak{m}}^{-1}\boldsymbol{I}_q - \boldsymbol{\mathcal{W}})\bigg)
\end{aligned}
\end{equation}
where $\boldsymbol{\mathcal{CCM}}=\dot{\boldsymbol{\mathcal M}} + \big\langle \boldsymbol{\mathcal M}\boldsymbol{\mathcal{A}} \big\rangle + 2\lambda \boldsymbol{\mathcal M}$. 

Second, a geometric controller modified from \cite{Lee2010} is employed as shown below. Given nominal states $\boldsymbol{x}^*=[\boldsymbol{p}_I^{*\top}\,\,\boldsymbol{v}_I^{*\top}\,\,\boldsymbol{r}_{IB}^{*\top}]^\top$ and nominal control $\boldsymbol{u}^*=[f_t^*/m\,\,\,\boldsymbol{\omega}_B^{*\top}]^\top$, a desired acceleration vector is obtained as
\begin{equation}
    \boldsymbol{a}_{des}=-\boldsymbol{K}_p\boldsymbol{e}_p - \boldsymbol{K}_v\boldsymbol{e}_v + \tfrac{f_t^*}{m}\boldsymbol{R}^*_{IB}\boldsymbol{e}_3
\end{equation}
where $\boldsymbol{K}_p=\operatorname{diag}(0.5,0.5,0.5)$ and $\boldsymbol{K}_v=\boldsymbol{I}_3$ are proportional and derivative gains for translational errors. Thrust input is determined by projecting it to the third axis of the desired body orientation such that $T_{geo}/m=\boldsymbol{a}_{des}^\top\boldsymbol{R}^*_{IB}\boldsymbol{e}_3$. Desired body-frame axes are then computed, where $\boldsymbol{z}_{B,des}=\boldsymbol{a}_{des}/\|\boldsymbol{a}_{des}\|$. Given the nominal yaw angle $\psi^*$, we find the heading vector on the $xy$-plane to be $\boldsymbol{x}_C=[\cos\psi^*\,\,\sin\psi^*\,\,0]^\top$. The remaining desired body axes can be found as $\boldsymbol{y}_{B,des}=\tfrac{\boldsymbol{z}_{B,des}^\wedge\boldsymbol{x}_C}{\|\boldsymbol{z}_{B,des}^\wedge\boldsymbol{x}_C\|}$ and $\boldsymbol{x}_{B,des}=\boldsymbol{y}_{B,des}^\wedge\boldsymbol{z}_{B,des}$. The desired rotation matrix can be computed as 
\begin{equation}\label{eqn: R_IB_construction}
    \boldsymbol{R}_{IB,des} = \begin{bmatrix}
        \boldsymbol{x}_{B,des} & \boldsymbol{y}_{B,des} & \boldsymbol{z}_{B,des}.
    \end{bmatrix}
\end{equation}
We then define our body rate control input based on geometric rotation error and the feedforward body rate such that 
\begin{equation}
    \boldsymbol{\omega}_{B,geo}=\boldsymbol{R}_{IB}^\top\boldsymbol{R}^*_{IB}\boldsymbol{\omega}_B^* - \boldsymbol{K}_r\boldsymbol{e}_R,
\end{equation}
where $\boldsymbol{e}_R=\frac{1}{2}(\boldsymbol{R}_{IB,des}^\top\boldsymbol{R}_{IB} - \boldsymbol{R}_{IB}^\top\boldsymbol{R}_{IB,des})^\vee$ and $\boldsymbol{K}_R=\operatorname{diag}(2.4,2.4,2.4)$. Combining the inputs gives the control $\boldsymbol{u}=[T_{geo}/m\,\,\,\boldsymbol{\omega}_{B,geo}^{\top}]^\top$.

Third, an acceleration-based uncertainty and disturbance estimator (UDE) is implemented with the RCCM controller to reduce the disturbance deviation. The disturbance estimation law is given by 
\begin{equation}\label{eqn: translational_UDE}
    \dot{\hat{\boldsymbol{f}_d}} = -\lambda_d\big(\hat{\boldsymbol{f}_d}-(m\dot{\boldsymbol{v}_I}-m\boldsymbol{g}_I-\boldsymbol{R}_{IB}\boldsymbol{e}_3\cdot f_t)\big),
\end{equation}
where $\hat{\boldsymbol{f}_d}$ is the estimated disturbance force vector and the UDE gain is chosen to be $\lambda_d=0.5$. The disturbance estimate is incorporated into differential flatness-based motion planning with $\boldsymbol{w}^*=[\hat{\boldsymbol{f}}_d^\top/m, \boldsymbol{0}]^\top$. 

Under differential flatness, the translational dynamics with disturbance can be expressed as
\begin{equation}\label{eqn: df_translation_dist}
m\ddot{\boldsymbol{p}}_I^* = m\boldsymbol{g}_I + \boldsymbol{R}_{IB}^*\boldsymbol{e}_3 f_t^* + \boldsymbol{f}_d,
\end{equation}
which yields the desired total force
\begin{equation}\label{eqn: df_total_force}
\boldsymbol{f}^* = m(\ddot{\boldsymbol{p}}_I^* - \boldsymbol{g}I) - \hat{\boldsymbol{f}}_d.
\end{equation}
Accordingly, the desired thrust magnitude and orientation are obtained as
\begin{equation}\label{eqn: df_thrust}
f_t^* = |\boldsymbol{f}^*|, \quad
\boldsymbol{z}_B^* = \frac{\boldsymbol{f}^*}{|\boldsymbol{f}^*|}.
\end{equation}
$\boldsymbol{R}_{IB}^*$ can be obtained from $\psi^*$ similar to \eqref{eqn: R_IB_construction}. Defining the jerk of nominal motion $\boldsymbol{j}=\dot{\boldsymbol{f}^*}/m$, $\boldsymbol{x}_C=[\cos\psi^*,\sin\psi^*,0]$ and $\boldsymbol{y}_C=[-\sin\psi^*,\cos\psi^*,0]$, the desired body angular rates $\boldsymbol{\omega_B^*}=[\omega_x,\,\omega_y,\,\omega_z]^\top$ can be found as 
\begin{equation}
\begin{aligned}
    \omega_x &= -{\boldsymbol{y}_B^*}^\top\boldsymbol{j}, \\ \omega_y &= {\boldsymbol{x}_B^*}^\top\boldsymbol{j}, \\ \omega_z &= \frac{1}{\boldsymbol{y}_C \times \boldsymbol{z}_B}
\left(
\dot{\psi}\,\boldsymbol{x}_C^\top \boldsymbol{x}_B
+ \omega_y\,\boldsymbol{y}_C^\top \boldsymbol{z}_B
\right).
\end{aligned}
\end{equation}
Therefore, we can construct the nominal motion as $\boldsymbol{x}^*=[\boldsymbol{p}_I^*,\,\dot{\boldsymbol{p}}_I^*,\,\operatorname{vec}(\boldsymbol{R}^*_{IB})]^\top$ and $\boldsymbol{u}^*=[f_t^*,\,\boldsymbol{\omega}_B^*]^\top$. The implementation framework is illustrated in Figure \ref{fig: UDE}.

\begin{figure}[htbp]
    \centering
    \includegraphics[width=\linewidth]{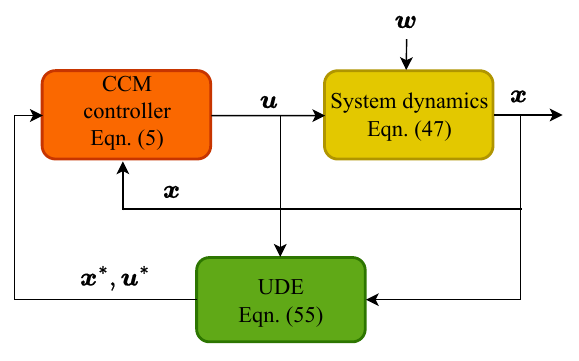}
    \caption{Closed-loop system with a CCM feedback controller and a UDE.}
    \label{fig: UDE}
\end{figure}

Numerical simulations are performed on the four test cases, all initialized from the same state with an initial error. All controllers achieve bounded tracking performance as illustrated in Fig. \ref{fig: 3D}. However, further examination of the of the output deviation norm $\|\boldsymbol{z}(t)-\boldsymbol{z}^*(t)\|$ in Fig. \ref{fig: tube} shows that both the CCM and geometric controllers exceed the prescribed bound of $\alpha \bar{w} = 1.268$, whereas the output $\boldsymbol{z}$ under the RCCM controller remains within the set defined in \eqref{eqn: tube_euclidean} throughout the simulation. 
\begin{figure}[htbp]
    \centering
    \includegraphics[width=\linewidth]{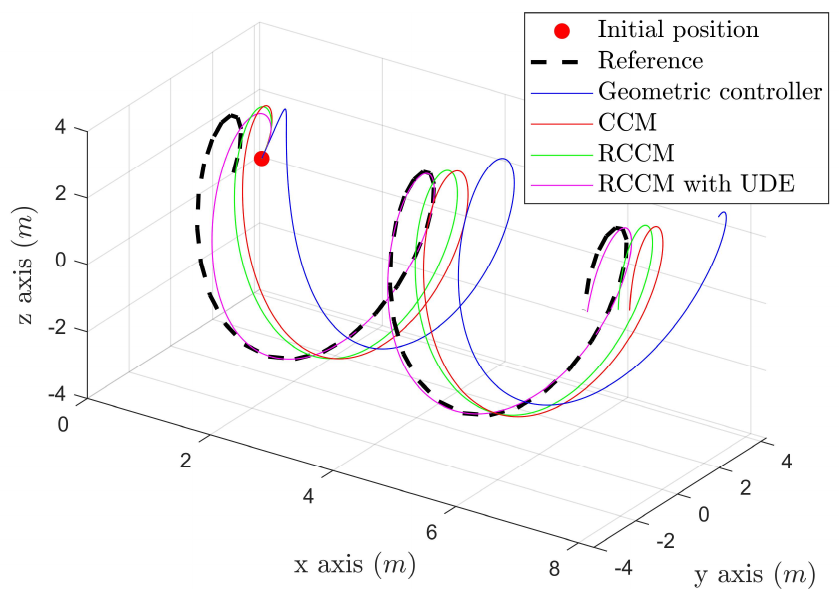}
    \caption{Quadrotor tracking a spiral trajectory under the four test cases.}
    \label{fig: 3D}
\end{figure}
\begin{figure}[htbp]
    \centering
    \includegraphics[width=\linewidth]{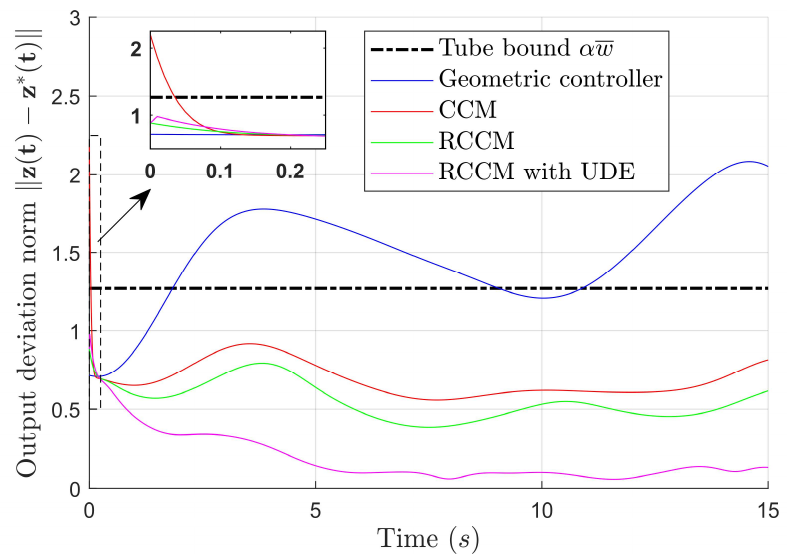}
    \caption{Output deviation norm under the four test cases.}
    \label{fig: tube}
\end{figure}

Moreover, the RCCM controller augmented with UDE exhibits the smallest output deviation, as persistent disturbances are effectively rejected after the transient phase, reducing disturbance deviation $|\boldsymbol{w}-\boldsymbol{w}^*|$. This improved disturbance attenuation directly contributes to enhanced tracking performance and system robustness compared to the RCCM controller. 

Fig.~\ref{fig: UDE_dist_est} illustrates the force disturbance estimation error throughout the simulation. During the initial $7$ seconds, the persistent disturbance is rapidly identified and compensated by the UDE, leading to a decline in estimation error as the observer converges. Beyond this period, the residual estimation error is dominated by sinusoidal components, which align with the injected time-varying disturbances. These components are not fully captured due to the finite bandwidth and inherent limitations of the observer, indicating a trade-off between disturbance rejection capability and estimator responsiveness.

\begin{figure}[htbp]
    \centering
    \includegraphics[width=\linewidth]{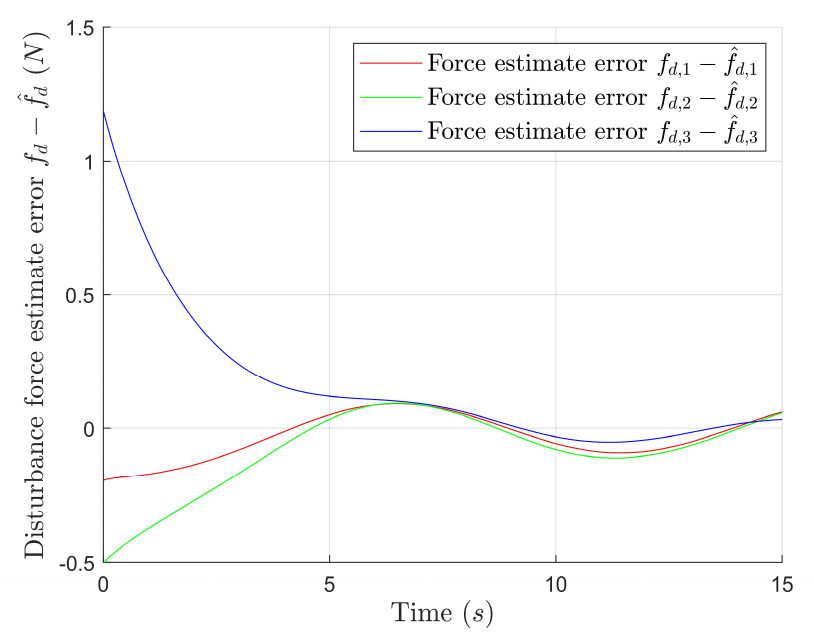}
    \caption{Disturbance estimation error of UDE.}
    \label{fig: UDE_dist_est}
\end{figure}

\section{Conclusion}
This paper extends the robust CCM approach to systems evolving on Lie groups, providing sufficient conditions for input–output stability. A learning-based framework is proposed to jointly synthesize the RCCM and the feedback controller. Simulation results on a quadrotor demonstrate superior performance in robustness under disturbances compared to geometric and CCM controllers. Future work will explore improved neural network architectures to reduce learning error and validate the proposed approach experimentally under real-world uncertainties.

\bibliography{citations}

@article{Singh2023,
author = {Singh, Sumeet and Landry, Benoit and Majumdar, Anirudha and Slotine, Jean-Jacques and Pavone, Marco},
issn = {0278-3649},
journal = {The International Journal of Robotics Research},
number = {9},
pages = {655--688},
publisher = {SAGE Publications Sage UK: London, England},
title = {{Robust feedback motion planning via contraction theory}},
volume = {42},
year = {2023}
}

@article{Faessler2017,
author = {Faessler, Matthias and Franchi, Antonio and Scaramuzza, Davide},
issn = {2377-3766},
journal = {IEEE Robotics and Automation Letters},
number = {2},
pages = {620--626},
publisher = {IEEE},
title = {{Differential flatness of quadrotor dynamics subject to rotor drag for accurate tracking of high-speed trajectories}},
volume = {3},
year = {2017}
}

@article{Brunke2022,
author = {Brunke, Lukas and Greeff, Melissa and Hall, Adam W and Yuan, Zhaocong and Zhou, Siqi and Panerati, Jacopo and Schoellig, Angela P},
issn = {2573-5144},
journal = {Annual Review of Control, Robotics, and Autonomous Systems},
number = {1},
pages = {411--444},
publisher = {Annual Reviews},
title = {{Safe learning in robotics: From learning-based control to safe reinforcement learning}},
volume = {5},
year = {2022}
}

@article{Zhao2024,
author = {Zhao, Pan and Guo, Ziyao and Cheng, Yikun and Gahlawat, Aditya and Kang, Hyungsoo and Hovakimyan, Naira},
issn = {2218-6581},
journal = {Robotics},
number = {7},
pages = {99},
publisher = {MDPI},
title = {{Guaranteed trajectory tracking under learned dynamics with contraction metrics and disturbance estimation}},
volume = {13},
year = {2024}
}

@incollection{Lee2003,
author = {Lee, John M},
booktitle = {Introduction to smooth manifolds},
pages = {1--29},
publisher = {Springer},
title = {{Smooth manifolds}},
year = {2003}
}

@inproceedings{Lee2010,
author = {Lee, Taeyoung and Leok, Melvin and McClamroch, N Harris},
booktitle = {49th IEEE conference on decision and control (CDC)},
isbn = {1424477468},
pages = {5420--5425},
publisher = {IEEE},
title = {{Geometric tracking control of a quadrotor UAV on SE (3)}},
year = {2010}
}

@inproceedings{Vang2020,
author = {Vang, Bee and Tron, Roberto},
booktitle = {2020 59th IEEE Conference on Decision and Control (CDC)},
isbn = {1728174473},
pages = {2006--2013},
publisher = {IEEE},
title = {{Global attitude control via contraction on manifolds with reference trajectory and optimization}},
year = {2020}
}

@inproceedings{Wu2024,
author = {Wu, Dongjun and Yi, Bowen and Manchester, Ian R},
booktitle = {2024 IEEE 63rd Conference on Decision and Control (CDC)},
isbn = {9798350316339},
pages = {3735--3740},
publisher = {IEEE},
title = {{Control contraction metrics on submanifolds}},
year = {2024}
}

@inproceedings{Sun2021,
author = {Sun, Dawei and Jha, Susmit and Fan, Chuchu},
booktitle = {conference on Robot Learning},
isbn = {2640-3498},
pages = {1519--1539},
publisher = {PMLR},
title = {{Learning certified control using contraction metric}},
year = {2021}
}

@inproceedings{Tsukamoto2021,
author = {Tsukamoto, Hiroyasu and Chung, Soon-Jo and Slotine, Jean-Jacques and Fan, Chuchu},
booktitle = {2021 60th IEEE Conference on Decision and Control (CDC)},
isbn = {166543659X},
pages = {2949--2954},
publisher = {IEEE},
title = {{A theoretical overview of neural contraction metrics for learning-based control with guaranteed stability}},
year = {2021}
}

@article{Manchester2017,
author = {Manchester, Ian R and Slotine, Jean-Jacques E},
issn = {0018-9286},
journal = {IEEE Transactions on Automatic Control},
number = {6},
pages = {3046--3053},
publisher = {IEEE},
title = {{Control contraction metrics: Convex and intrinsic criteria for nonlinear feedback design}},
volume = {62},
year = {2017}
}

@article{Manchester2018,
author = {Manchester, Ian R and Slotine, Jean-Jacques E},
issn = {2475-1456},
journal = {IEEE Control Systems Letters},
number = {3},
pages = {333--338},
publisher = {IEEE},
title = {{Robust Control Contraction Metrics: A Convex Approach to Nonlinear State-Feedback ${H}^\infty $ Control}},
volume = {2},
year = {2018}
}

@article{Tsukamoto2020,
author = {Tsukamoto, Hiroyasu and Chung, Soon-Jo},
issn = {2475-1456},
journal = {IEEE Control Systems Letters},
number = {1},
pages = {211--216},
publisher = {IEEE},
title = {{Neural contraction metrics for robust estimation and control: A convex optimization approach}},
volume = {5},
year = {2020}
}

@article{Zhao2022,
author = {Zhao, Pan and Lakshmanan, Arun and Ackerman, Kasey and Gahlawat, Aditya and Pavone, Marco and Hovakimyan, Naira},
issn = {2377-3766},
journal = {IEEE Robotics and Automation Letters},
number = {2},
pages = {5528--5535},
publisher = {IEEE},
title = {{Tube-certified trajectory tracking for nonlinear systems with robust control contraction metrics}},
volume = {7},
year = {2022}
}

@article{Dawson2023,
author = {Dawson, Charles and Gao, Sicun and Fan, Chuchu},
doi = {10.1109/TRO.2022.3232542},
issn = {19410468},
journal = {IEEE Transactions on Robotics},
number = {3},
pages = {1749--1767},
publisher = {IEEE},
title = {{Safe Control With Learned Certificates: A Survey of Neural Lyapunov, Barrier, and Contraction Methods for Robotics and Control}},
volume = {39},
year = {2023}
}

@misc{Lo2026,
author = {Lo, Yi Lok and Qian, Longhao and Liu, Hugh H.-T.},
title = {{Neural Robust Control on Lie Groups Using Contraction Methods}},
url = {https://github.com/loyilok515/CDC2026},
year = {2026}
}

@article{Lohmiller1998,
author = {Lohmiller, Winfried and Slotine, Jean-Jacques E},
issn = {0005-1098},
journal = {Automatica},
number = {6},
pages = {683--696},
publisher = {Elsevier},
title = {{On contraction analysis for non-linear systems}},
volume = {34},
year = {1998}
}
\end{document}